\documentclass[11pt]{article}
\usepackage{epsfig}
\usepackage{amsfonts}
\usepackage{amssymb}
\usepackage{amsthm}
\usepackage{amstext}
\usepackage{amsmath}
\usepackage{xspace}
\usepackage{graphicx}
\usepackage{graphics}
\usepackage{colordvi}
\usepackage{color}
\usepackage{psfrag}
\usepackage{subfigure}
\usepackage{algorithm, algorithmic}
\usepackage{wrapfig}
\usepackage{url}
\usepackage[margin=1in]{geometry}
\usepackage{fullpage}
\usepackage{multirow}
\usepackage{hyperref}
\usepackage{verbatim}

\newcommand{\yi}{{\sc Yes-Instance}\xspace}
\renewcommand{\ni}{{\sc No-Instance}\xspace}

\newcommand{\size}[1]{\ensuremath{\left|#1\right|}}

\newcommand{\NP}{\mbox{\sf NP}}
\newcommand{\APX}{\mbox{\sf APX}}

\newcommand{\polylog}[1]{\mathrm{polylog(#1)}}
\newcommand{\DTIME}{\mbox{\sf DTIME}}

\newcommand{\opt}{\mbox{\sf OPT}}

\newcommand{\set}[1]{\left\{ #1 \right\}}
\newcommand{\sse}{\subseteq}

\newcommand{\tset}{{\mathcal T}}
\newcommand{\iset}{{\mathcal{I}}}
\newcommand{\pset}{{\mathcal{P}}}

\newcommand{\lset}{{\mathcal{L}}}

\newcommand{\cset}{{\mathcal{C}}}

\newcommand{\xset}{{\mathcal{X}}}

\newtheorem{theorem}{Theorem}[section]
\newtheorem{lemma}[theorem]{Lemma}
\newtheorem{observation}[theorem]{Observation}
\newtheorem{corollary}[theorem]{Corollary}
\newtheorem{claim}[theorem]{Claim}

\newtheorem{openproblem}{Open Problem}

\def\square{\vbox{\hrule height.2pt\hbox{\vrule width.2pt height5pt \kern5pt
\vrule width.2pt} \hrule height.2pt}}

\theoremstyle{definition}\newtheorem{definition}[theorem]{Definition}

\newenvironment{prog}[1]{
\begin{minipage}{5.8 in}
{\sc\bf #1}
\begin{enumerate}}
{
\end{enumerate}
\end{minipage}
}

\renewcommand{\phi}{\varphi}

\newcommand{\poly}{\operatorname{poly}}

\newcommand{\reals}{{\mathbb R}}

\newcommand{\R}{\ensuremath{\mathbb R}}

\newcommand{\expect}[2]{\text{\bf E}_{#1}\left [#2\right]}

\newcommand{\pr}[2]{\text{\bf Pr}_{#1}\left [#2\right ]}

\newcommand{\cI}{\mathcal{I}}
\newcommand{\cC}{\mathcal{C}}
\newcommand{\cP}{\mathcal{P}}

\newcommand{\draftonly}[1]{}

\newcommand{\fullonly}[1]{}
\newcommand{\sodaonly}[1]{#1}

\def\danupon#1{}
\def\parinya#1{}
\def\khaled#1{}
\def\he#1{}

\usepackage[small,compact]{titlesec}
\usepackage[footnotesize,it]{caption}
\newcommand{\squishlist}{
 \begin{list}{$\bullet$}
  { \setlength{\itemsep}{0pt}
     \setlength{\parsep}{3pt}
     \setlength{\topsep}{3pt}
     \setlength{\partopsep}{0pt}
     \setlength{\leftmargin}{1.5em}
     \setlength{\labelwidth}{1em}
     \setlength{\labelsep}{0.5em} } }
\newcommand{\squishend}{
  \end{list}  }

\newcounter{Lcount}
\newcommand{\squishlisttwo}{
\begin{list}{\arabic{Lcount}. }
{ \usecounter{Lcount} \setlength{\itemsep}{0pt}
\setlength{\parsep}{0pt} \setlength{\topsep}{0pt}
\setlength{\partopsep}{0pt} \setlength{\leftmargin}{2em}
\setlength{\labelwidth}{1.5em} \setlength{\labelsep}{0.5em} } }

\newcommand{\squishendtwo}{
\end{list} }

\newcommand{\Item}{{\bf I}}
\newcommand{\Consumer}{{\bf C}}

\newcommand{\vol}{\mathsf{vol}}

\usepackage{xcolor,setspace}

\title{Geometric Pricing: How Low Dimensionality Helps in Approximability}
\date{}

\author{
Parinya Chalermsook\thanks{IDSIA, Lugano, Switzerland, Email:  {\tt parinya@uchicago.edu}. Work partially done while the author was at the University of Chicago, Chicago, IL, USA, and Max-Planck-Institut f\"ur Informatik, Saarbr\"ucken, Germany.}
\and Khaled Elbassioni\thanks{Max-Planck-Institut f\"ur Informatik,
Saarbr\"ucken, Germany. Email: {\tt
\{elbassio,hsun\}@mpi-inf.mpg.de}}
\and Danupon Nanongkai\thanks{Theory and Applications of Algorithms Research Group, University of Vienna, Vienna, Austria. Email: {\tt
danupon@gmail.com}. Work partially done while the author was at Georgia Institute of Technology, Atlanta, GA, USA and Max-Planck-Institut f\"ur Informatik, Saarbr\"ucken, Germany.}
\and He Sun\addtocounter{footnote}{-3}\footnotemark
}

\begin{document}

\begin{titlepage}

\maketitle

\thispagestyle{empty}



\begin{abstract}
Consider the following toy problem. There are $m$ rectangles and $n$ points on the plane. Each rectangle $R$ is a consumer with budget $B_R$, who is interested in purchasing the cheapest item (point) inside R, given that she has enough budget. Our job is to price the items to maximize the revenue. This problem can also be defined on higher dimensions. We call this problem the {\em geometric pricing} problem.

In high dimensions, the above problem is equivalent to the {\em unlimited-supply profit-maximizing pricing} problem, which has been studied extensively in approximation algorithms and algorithmic game theory communities. Previous studies suggest that the latter problem is too general to obtain a sub-linear approximation ratio (in terms of the number of items) even when the consumers are restricted to have very simple purchase strategies.

In this paper, we study a new class of problems arising from a geometric aspect of the pricing problem. It intuitively captures typical real-world assumptions that have been widely studied in marketing research, healthcare economics, etc. It also helps classify other well-known pricing problems, such as the {\em highway pricing} problem and the {\em graph vertex pricing} problem on planar and bipartite graphs. Moreover, this problem turns out to have close connections to other natural geometric problems such as the geometric versions of the {\em unique coverage} and {\em maximum feasible subsystem} problems.

We show that the low dimensionality arising in this pricing problem does lead to improved approximation ratios, by presenting sublinear-approximation algorithms for two central versions of the problem:  {\em unit-demand uniform-budget min-buying} and {\em single-minded} pricing problems. Our algorithm is obtained by combining algorithmic pricing and geometric techniques. These results suggest that considering geometric aspect might be a promising research direction in obtaining improved approximation algorithms for such pricing problems. To the best of our knowledge, this is one of very few problems in the intersection between geometry and algorithmic pricing areas. Thus its study may lead to new algorithmic techniques that could benefit both areas.
\end{abstract}

\thispagestyle{empty}





\thispagestyle{empty}

\end{titlepage}

\newpage

\setcounter{page}{1}

\section{Introduction}

This paper studies a geometric version of two central {\em unlimited-supply pricing} problems. We are given a set $\iset$ of $n$ consumers and a set $\cset$ of $m$ items.
Every item ${\bf I}\in\iset$ is represented by a point ${\bf I} = ({\bf I}[1],\ldots, {\bf I}[d]) \in \R^d_{\geq 0}$, where $\R_{\geq 0}$ denotes the set of non-negative reals and ${\bf I}[j]$ expresses the quality of item ${\bf I}$ in the $j$-th attribute. Every consumer ${\bf C}\in\cset$ is represented by a point ${\bf C}= ({\bf C}[1],\ldots, {\bf C}[d]) \in \R^d_{\geq 0}$, where ${\bf C}[j]$ is the criterion of consumer ${\bf C}\in\mathcal{C}$ in the $j$-th attribute. Each consumer $\Consumer$ is additionally equipped with budget $B_{\Consumer}\in\mathbb{R}_{\geq 0}$ and a consideration set
\begin{equation}\label{eq:SC for UDP-MIN}
S_\Consumer=\set{{\bf I}: {\bf I}[j] \geq {\bf C}[j], \mbox{for all } 1\leq j \leq d}.
\end{equation}
\begin{wrapfigure}{r}{0.3\textwidth}
\center
  \vspace{-.9 cm}
  \includegraphics[width=1.1\linewidth, clip=true, trim= 0.5cm 1cm 1cm 0.8cm]{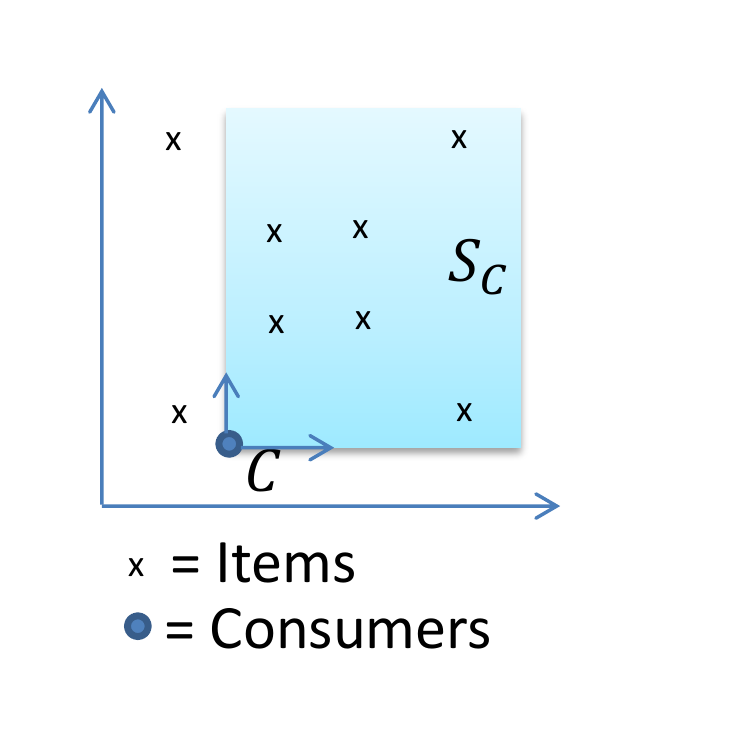}\\
  \caption{Problem visualization}\label{fig:visualization}\vspace{-.3cm}
\end{wrapfigure}
In the {\em $d$-dimensional uniform-budget unit-demand min-buying pricing problem} ($d$-{\sf UUDP-MIN}), once we assign prices to items, each consumer ${\bf C}$ will buy the cheapest item ${\bf I}$ in $S_{\Consumer}$ if the price of item  ${\bf I}$ is at most $B_{\Consumer}$. In the {\em $d$-dimensional single-minded pricing problem} ($d$-{\sf SMP}), consumer ${\bf C}$ will buy the {\em all} items in $S_{\Consumer}$ if the total price of those items  is at most $B_{\Consumer}$.
The objective is to set the price of items in $\iset$ in order to maximize the revenue. That is, we want to find $p: \iset\rightarrow \reals_{\geq 0}$ that maximizes $\sum_{{\bf C}\in \cset, \min_{\Item\in S_\Consumer} p(\Item)\leq B_\Consumer} \min_{\Item\in S_\Consumer} p(\Item)$ in the case of $d$-{\sf UUDP-MIN} and $\sum_{{\bf C}\in \cset, \sum_{\Item\in S_\Consumer} p(\Item)\leq B_\Consumer} \sum_{\Item\in S_\Consumer} p(\Item)$ in the case of $d$-{\sf SMP}.
Fig.~\ref{fig:visualization} illustrates the problem: Each item corresponds to a point in the plane. The consideration set of each consumer $\Consumer$ is represented by an (unbounded) axis-parallel rectangle with point ${\bf C}$ as a lower-left corner.



The above problems when $d$ is unbounded (called {\sf UUDP-MIN} and {\sf SMP}) have been widely studied recently
(e.g.,
\cite{Rusmevichientong03,GuruswamiHKKKM05,RusmevichientongRG06,BriestK11,ChalermsookPricing,AggarwalFMZ04,BalcanB07})
and are known to be $O(\log m)$-approximable~\cite{AggarwalFMZ04}; so we have a reasonable approximation guarantee when there are not many consumers.
However, in many cases, one would expect the number of consumers to be much larger than the number of items~$n$. In this case, we are still stuck at the trivial $O(n)$ approximation ratio, and there are evidences that suggest that getting a sub-linear approximation ratios might be impossible: Unless $\NP \subseteq \DTIME(n^{\poly  \log n})$, these problem are hard to approximate within a $2^{\log^{1-\epsilon} n}$ for any constant $\epsilon>0$~\cite{ChalermsookPricing}. Moreover, assuming a stronger (but still plausible) assumption, these problems are hard to approximate to within a factor of $n^{\epsilon}$ for some $\epsilon >0$ \cite{BriestK11}.


Motivated by various types of assumptions, the pricing problems with special structures have been studied (e.g., when there is a {\em price-ladder constraint}~\cite{Rusmevichientong03,AggarwalFMZ04,RusmevichientongRG06,RusmevichientongRG06-2,BriestK11}, consideration sets are small~\cite{BalcanB07,BriestK11} or consideration sets correspond to paths on graphs~\cite{BalcanB07,ElbassioniSZ07,GrandoniR10,ElbassioniRRS09,DBLP:conf/icalp/GamzuS10}). In these cases, better approximation ratios are usually possible.

In this paper we consider the geometric structure of pricing problems arising naturally from real-world scenarios, which turns out to be quite general. Our motivation is two-fold: We hope that the geometric structures will lead to better approximation algorithms, and we found these problems interesting on their own as they have connections to other pricing and geometric problems.

Our problems are motivated by the following simple observation on the consumers' behavior.
Consider a setting where we sell cars. If a consumer has car $A$ with horse power $130$HP in her consideration set, she would not mind buying car $B$ with horse power $150$HP. Maybe she does not want $B$ because it is less energy-efficient or has lower reputation. But, if we list {\em all} attributes of the cars that people care about and it happens that $B$ is not worse than $A$ in all other aspects, then $B$ should also be in the list.


In particular, instead of looking at a full generality where each consumer $\Consumer$ considers any set of items $S_\Consumer$, it is reasonable to assume that each consumer has some criterion in mind for each attribute of the cars, and her consideration set consists of any car that passes all her criteria, i.e.  consumers judge items according to their attributes.
This natural assumption  has been a model of study in other fields such as marketing research, healthcare economics and urban planning. It is referred to as the {\em attribute-based screening process}. In particular, using criteria to define consideration sets as in Eq.~\eqref{eq:SC for UDP-MIN} is called {\em conjunctive screening rule}. Besides being natural, this assumption has been supported by a number of studies where it is concluded that consumers typically use a conjunctive screening rule in obtaining their consideration sets (see further detail in Section~\ref{sec:relatedwork}).

It is also interesting that $d$-{\sf SMP} captures many previously studied problems as special cases. For example, $2$-{\sf SMP} generalizes the highway pricing problem~\cite{GuruswamiHKKKM05,BalcanB07,ElbassioniSZ07,GrandoniR10} and thus our algorithmic results on $2$-{\sf SMP} can immediately be applied to this problem. Moreover, $3$-{\sf SMP} generalizes the upward case of the tollbooth pricing problem \cite{ElbassioniRRS09,KortsarzRR11} as well as the graph vertex pricing problem on planar graphs \cite{BalcanB07,ChalermsookKLN11}. $4$-{\sf SMP} generalizes the unlimited-supply version of the {\em exhibition} problem \cite{ChristodoulouEF10}, the graph vertex pricing problem on bipartite graphs \cite{BalcanB07,KhandekarKMS09}, and the ``rectangle version'' of the {\em unique coverage problem} ({\sc UC}) \cite{DemaineFHS08}, which are the geometric variants of {\sc UC} studied recently in~\cite{ErlebachL08,ItoNOOUUU12}.



Moreover, {\sf SMP} is a special case of the {\em maximum feasible subsystem with 0/1 coefficients} problem ({\sc Mrfs}) \cite{ElbassioniRRS09}. Elbassioni et~al. \cite{ElbassioniRRS09} showed that a very special geometric version of {\sc Mrfs} (the ``interval version'') admits much better approximation ratios than the general one. A geometric {\sc Mrfs} can be seen as a special case of ``$2$-{\sc Mrfs}'' in our terminologies, and it is thus interesting whether ``$d$-{\sc Mrfs}'' is easier than general {\sc Mrfs} for other values of $d$. Our geometric {\sf SMP} is a special case of $d$-{\sc Mrfs}. Thus, solving $d$-{\sf SMP} serves as the first step towards solving  $d$-{\sc Mrfs}.

%


\subsection{Our Results and Techniques}

We show that geometric structures lead to breaking the linear-approximation barrier: While the pricing problems are likely to be hard to approximate within a factor of $n^{1-\epsilon}$ in the general cases, we obtain an $o(n)$-approximation algorithms in the geometric setting, as follows.

\begin{theorem}\label{thm:intro thm udp smp}
For any $d>0$, there is an $\tilde O_d\left(n^{1-\epsilon(d)}\right)$-approximation algorithm for $d$-{\sf UUDP-MIN} and $d$-{\sf SMP} where function $\epsilon(d): = \frac{1}{4^{d-1}}$ and $\tilde O_d$ treats $d$ as a constant and hides a $\mathrm{polylog}(n)$ factor.
\end{theorem}

The essential idea behind our algorithm is to partition the problem instance into sub-instances without decreasing the optimal revenue (we call this {\em consideration-preserving decomposition}). This is done by using Dilworth's Theorem (partitioning items into chains and anti-chains) and epsilon-nets to find subsets of items satisfying certain structural properties. Subsequently, we show that the dimensions of these sub-instances can be reduced through the notion of {\em consideration-preserving embedding}. In the end of our algorithm, we are left with a sub-linear number of sub-instances, each of which can be solved almost optimally in polynomial time. Returning the best solution among the solutions of these sub-instances guarantees a sub-linear approximation ratio.

The spirit of our technique is in some sense in a similar flavor to Chan's algorithm \cite{Chan12} which computes a {\em conflict-free} coloring of $d$-dimensional points (w.r.t. rectangle ranges) using $O(n^{1-0.632/(2^{d-3}-0.368)})$ colors. In particular, in the 2-dimensional cases of both our geometric pricing and Chan's conflict-free coloring problems, the upper bounds of $O(\sqrt{n})$ can be obtained by a simple application of Dilworth's theorem (Ajwani et al. \cite{AjwaniEGR12} obtained a better bound in this case for the latter problem). However, the techniques of the two results are different in higher dimensions.

\paragraph{QPTASs} We also obtain {\sf QPTAS}s for $2$-{\sf UUDP-MIN} and $2$-{\sf SMP}. We present this in Appendix~\ref{sec: qptas 2 udp} and \ref{sec:2-SMP}. These results, together with a widely-believed assumption that the existence of a {\sf QPTAS} for any problem implies that {\sf PTAS} exists for the same problem (e.g., \cite{BansalCES06,ElbassioniSZ07}), imply that the value of $\epsilon(d)$ in Theorem~\ref{thm:intro thm udp smp} could be improved slightly to $1/4^{d-2}$.
As a by-product of these results, we show a {\sf QPTAS} for $2$-{\sf SMP} which subsumes the previous {\sf QPTAS} for highway pricing \cite{ElbassioniSZ07}.

\paragraph{Hardness} We also study the hardness of approximation of our problems. We show that $3$-{\sf UUDP-MIN} and $2$-{\sf SMP} are {\sf NP}-hard, and $4$-{\sf UUDP-MIN} and $4$-{\sf SMP} are {\sf APX}-hard. Hence, our problem is already non-trivial for small $d$.
Our hardness proofs establish a cute connection between our problem and the vertex cover problem on graphs of low {\em order dimensions}~\cite{Schnyder89,DBLP:conf/soda/Schnyder90}.
%
%
%
Moreover, we show that the hardness of our problem tends to increase as we increase $d$, and the whole generality is captured when $d=n$. In particular, we show that when the dimension is sufficiently high (i.e. $d \geq \log^2 n$), the problems are hard to approximate to within a factor of $d^{1/4-\epsilon}$ for any $\epsilon>0$. 
Table~\ref{table:current status} concludes our results for $d$-{\sf UUDP-MIN} and $d$-{\sf SMP}.


\begin{table}
\begin{center}
\footnotesize
\begin{tabular}{|c|c|c|c|c|c|c|}
  \hline
  {\bf Problem} & & $\mathbf{d=1}$ & $\mathbf{d=2}$ & $\mathbf{d=3}$ & $\mathbf{d=4}$ & {\bf large $\mathbf{d}$  \{range \} }\\
  \hline
  \multirow{2}{*}{$d$-{\sf UUDP-MIN}} & Upper bound & Polytime & {\sf QPTAS} &  & & $n^{1-\frac{1}{4^{d-1}}}$ \{constant $d$\}\\
  & Lower bound & & & {\sf NP}-hard &  {\sf APX}-hard & $d^{\frac{1}{4}-\epsilon}$ \{$d=\omega(\log n)$\} \\
  \hline
  \multirow{2}{*}{$d$-{\sf SMP}} & Upper bound & Polytime & {\sf QPTAS} &  & & $n^{1-\frac{1}{4^{d-1}}}$ \{constant $d$\}\\
  & Lower bound & & {\sf NP}-hard & & {\sf APX}-hard & $d^{\frac{1}{4}-\epsilon}$ \{$d=\omega(\log n)$\} \\
  \hline
\end{tabular}
\end{center}
\vspace{-.5cm}
\caption{Results of $d$-{\sf UUDP-MIN} and $d$-{\sf SMP} for small values of $d$.}
\label{table:current status}
\vspace{-.3cm}
\end{table}

\subsection{Related Work}\label{sec:relatedwork}

Rusmevichientong et al.~\cite{Rusmevichientong03,RusmevichientongRG06,RusmevichientongRG06-2} defined the {\em non-parametric multi-product
pricing problem}, motivated by the possibility that the data about consumers' preferences and budgets can be predicted based on previous data, which can be gathered and mined by web sites designed for this purpose, e.g., \cite{HaublTrifts00,RusmevichientongRG06-2}. This problem is what we call uniform-budget unit-demand pricing problem ({\sf UUDP}). Rusmevichientong et al. proposed many decision rules such as min-buying, max-buying and rank-buying and showed that {\sf UUDP-MIN} allows a polynomial-time algorithm, assuming the {\em price-ladder constraint}, i.e., a predefined total order on the prices of all products. Aggarwal et al.~\cite{AggarwalFMZ04} later showed that the price ladder constraint also leads to a $4$-approximation algorithm for the max-buying case, even in the case of limited supply.

We note that the price ladder constraint is closely related to our notion of attributes in the following sense. It can be shown that $1$-{\sf UUDP-MIN} satisfies the price ladder constraint (this is the reason we can solve it in polynomial time). Moreover, although $2$-{\sf UUDP-MIN} does not satisfy this constraint, it {\em partially} satisfies the constraint in the sense that if one item is better than another item in all attributes then we can assume that it has a higher price. This property plays an important role in obtaining {\sf QPTAS} for $2$-{\sf UUDP-MIN} and also holds for general $d$.

Other variants defined later include non-uniform and utility-maximizing unit-demand, single-minded ({\sf SMP}), tollbooth and highway models~\cite{AggarwalFMZ04,GuruswamiHKKKM05}.
These problems were later found to have important connections to
algorithmic mechanism design~\cite{AggarwalH06,BalcanB05,GuruswamiHKKKM05} and online pricing
problems~\cite{BalcanB07,BlumH05}.
As we mentioned in the introduction, many problems can be approximated within the factor of $O(\log m + \log n)$ and $O(n)$, and these seem to be tight.

The observation that consumers make decisions based on attributes has been used in  other areas outside computer science. For example, most pricing models are captured by the {\em two-stage consider-then-choose} model (e.g., \cite{Gensch87,Payne82,PayneBJ88,GilbrideAllenby,HaublTrifts00,JedidiK05,HauserTEBS10,LiuArora2011}) in marketing research: Each consumer first screens out some undesirable items ({\em screening process}) and is left with the consideration set which is used to make a final decision. Pricing problems such as {\sf UUDP-MIN} are the case where consideration sets are arbitrary (as defined in, e.g. \cite{Shocker91,HauserW90}) while the final decision is simplified to, e.g., buying the cheapest item.

The idea of using the consideration sets defined from attributes is called {\em attribute-based screening process}~\cite{GilbrideAllenby} in marketing research where it was shown to be a rational choice for trading off between accuracy and cognitive effort~\cite{Bettman79,BettmanJP90,BettmanPark80,Shugan80}. Our model is equivalent to the attribute-based screening process  with {\em conjunctive screening rules} (e.g., \cite{GilbrideAllenby,LiuArora2011}). This type of rules was justified by many studies that it is what consumers typically use when making decisions (e.g., \cite{Bettman79,GilbrideAllenby,HauserTEBS10}).


\section{Sub-linear Approximation Algorithm (Proof of Theorem~\ref{thm:intro thm udp smp})}\label{sec:d udp min}

To simplify the presentation, we present the algorithm for $d$-{\sf UUDP-MIN} in this section. The algorithm for $d$-{\sf SMP} is almost identical. Let $\cset$ and $\iset$ be the set  of points in $\mathbb{R}^d$, where every consumer $\Consumer\in{\cal C}$ has budget $B_{\Consumer}$ and consideration set $S_{\Consumer}$ which is specified by coordinates of the input point. For any subset $\mathcal{C}'\subseteq\cset$ and ${\cal I}'\subseteq \iset$, let
$\mathcal{P}(\cset',\iset')$ be the $d$-{\sf UUDP-MIN} problem with input $\cset'$ and $\iset'$. Moreover, for any $\cset'$ and $\iset'$, we use $\opt(\cset', \iset')$ to express the optimal revenue of the instance $(\cset', \iset')$. At a high level, our algorithm proceeds in an inductive manner and obtains a solution of $d$-{\sf UUDP-MIN} problem by invoking the algorithms for $(d-1)$-{\sf UUDP-MIN} and $1$-{\sf UUDP-MIN} as a subroutine. Our result is summarized in the following theorem.

\begin{theorem}
\label{thm: dimension reduction for UDP}
For any $\epsilon \in (0,1]$, if there is an $\tilde O_d(n^{1-\epsilon})$-approximation algorithm for $(d-1)$-{\sf UUDP-MIN} then there is an $\tilde O_d(n^{1-\epsilon/4})$-approximation algorithm for $d$-{\sf UUDP-MIN} as well.
\end{theorem}

Theorem~\ref{thm:intro thm udp smp} then follows from the fact that $1$-{\sf UUDP-MIN} can be solved optimally in polynomial time (see Appendix~\ref{sec:one udp min}\fullonly{ and its generalization in Section~\ref{sec: algorithms for small d}}). As we noted earlier, it can be improved slightly since $2$-{\sf UUDP-MIN} admits {\sf QPTAS} \sodaonly{(see Appendix~\ref{sec: qptas 2 udp})}\fullonly{(see Section~\ref{subsection: qptas minbuying})}.
%

\subsection{Consideration-preserving Decomposition}

Our algorithm partitions the input instance into many subinstances and tries to collect the profit from some of them. The notion of consideration-preserving decomposition, defined below, allows us to do so without losing revenue.

\begin{definition}\label{def:decomposition} We call a collection $\set{(\cset'_1,\iset'_1),\ldots, (\cset'_k, \iset'_k)}$  a {\em consideration-preserving decomposition} of the problem $(\cC, \cI)$ if and only if for any $\Consumer \in \cset$ and $\Item \in S_{\Consumer}$, there exists (not necessarily unique) $i$ such that $\Consumer \in \cset'_i$ and $\Item \in \iset'_i$.
\end{definition}

By definition, for any consumer $\Consumer$ and item $\Item$ the fact that consumer $\Consumer$ considers item $\Item$ is preserved by at least one instance $(\cset'_i, \iset'_i)$. The following lemma says that this decomposition preserves the total revenue.

\begin{lemma}
\label{lemma: UDP decomposition}
For any consideration-preserving decomposition $\left\{(\cset'_1,\iset'_1), \ldots, (\cset'_k, \iset'_k)\right\}$ of $(\cset, \iset)$, it holds that
$\sum_{i=1}^k \opt(\cset'_i, \iset'_i)\geq \opt(\cset, \iset)
\,.$
Moreover, any price function for $\pset(\cset'_i, \iset'_i)$ can be extended to a price function for the original problem $\pset(\cset, \iset)$ that gives revenue at least $\opt(\cset'_i, \iset'_i)$.
\end{lemma}

This is simply by applying the optimal price function of one problem to the other (see Appendix~\ref{sec: proof of UDP decomposition lemma} for the full proof). In the rest of our discussion, we mainly use two different types of consideration-preservation decomposition, as explained in the following observation.

\begin{observation}\label{observation: conseration preserving decomposition}
Given an input instance $(\cset',\iset')$, let $\cset' = \bigcup_{i=1}^k \cset'_i$. Then $\{(\cset'_1,\iset')$, $\ldots$, $(\cset'_k, \iset')\}$ is a consideration-preserving decomposition of $(\cset', \iset')$. Similarly, if $\iset' = \bigcup_{i=1}^k \iset'_i$, then we have that $\set{(\cset',\iset'_1), \ldots, (\cset', \iset'_k)}$ is a consideration-preserving decomposition of $(\cset', \iset')$.
\end{observation}

\subsection{Algorithm}

%

\begin{wrapfigure}{r}{0.5\textwidth}
  \vspace{-1.5cm}
  \center\includegraphics[width=\linewidth]{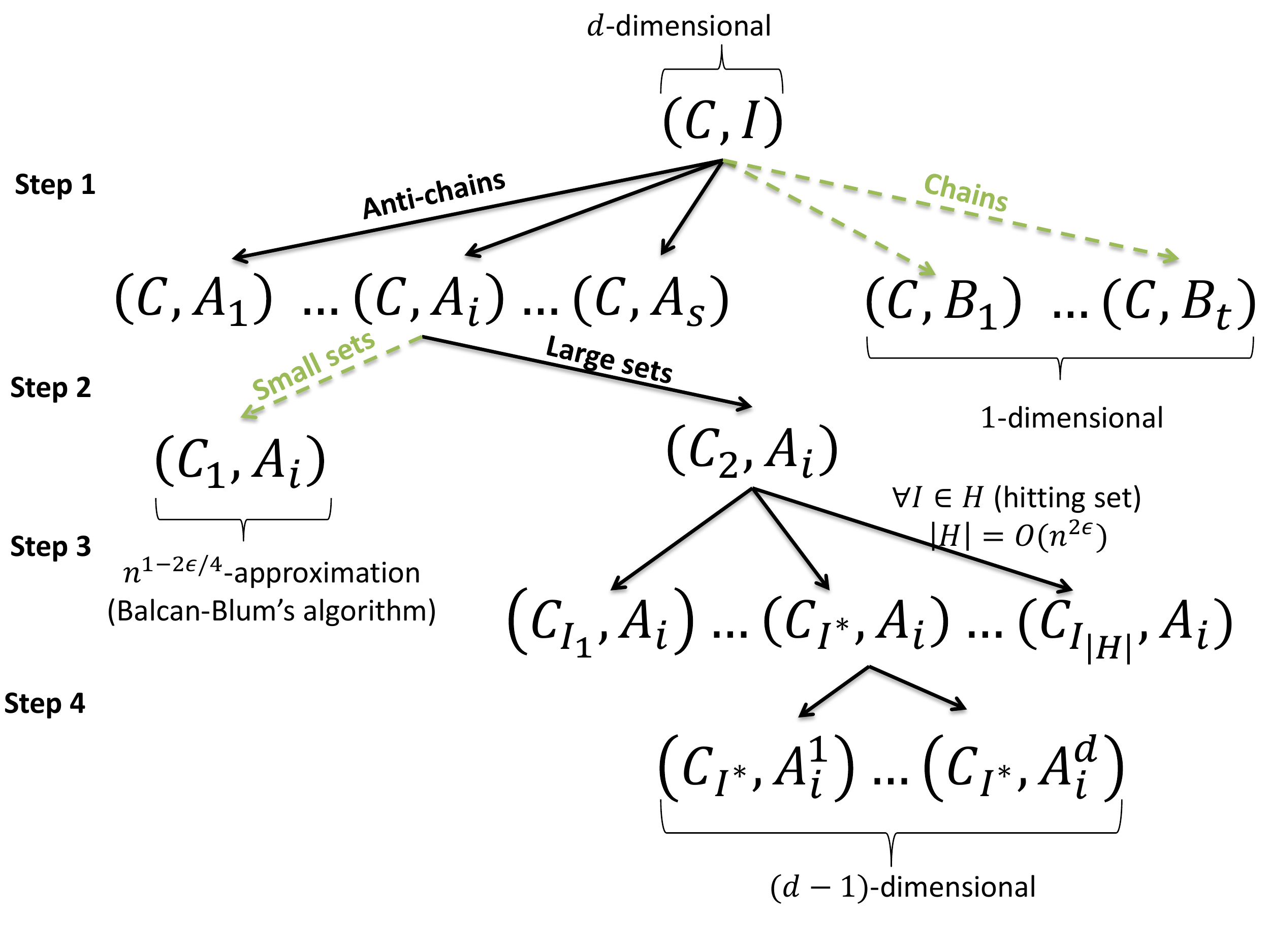}
  \vspace{-.5cm}
  \caption{Decomposition overview}\label{fig:overview}
  \vspace{-.5cm}
\end{wrapfigure}
At a high level, the algorithm proceeds in four steps where each step involves consideration-preserving decomposition (see Fig.~\ref{fig:overview} for an overview). In Step 1, we partition $\iset$ into different subsets where
 every subset satisfies certain properties, i.e. the elements in each subset either form a chain or an antichain. The problem on those subsets
  in which elements form a chain can be solved easily, and we deal with the antichains in later steps. In Step 2, we partition consumers in $\cset$ into two types, those with large and small consideration sets. We use the algorithm of \cite{BalcanB07,BriestK11} to deal with consumers with small consideration sets and handle the rest consumers in later steps. In Step 3, we find a subset of items, i.e. a ``hitting set'',  and partition consumers further into several sets. Each set of consumers has the following property: There is some item desired by all consumers in the set. Using this property, we show in Step 4 that the problem can be further partitioned into a few problems where each of them can be viewed as a $(d-1)$-{\sf UUDP-MIN} problem. (We call this a ``consideration-preserving embedding''.)




\paragraph{\underline{Step 1:} Partitioning items into chains and antichains}

Let $(\cset, \iset)$ be an input of $d$-{\sf UUDP-MIN}. First we define a partially ordered set $(\iset, \leq)$ on the item set as follows. We say that $\Item_1 \leq \Item_2$ if and only if $\Item_1$ has a lower quality than $\Item_2$ in every attribute, i.e. $\Item_1[d'] \leq \Item_2[d']$ for all $d' \in [d]$.
We say that a subset $\iset' \subseteq \iset$ is a chain if $\iset'$ can be written as $\iset'= \set{\Item_1,\ldots, \Item_z}$ such that $\Item_j \leq \Item_{j+1}$ for all $j\in[z-1]$. We say that $\iset' \subseteq \iset$ is an antichain if and only if for any pair of items $\Item, \Item' \in \iset'$, neither $\Item \leq \Item'$ nor $\Item' \leq \Item$.

\begin{lemma}
\label{lemma: chain decomposition}
For any $\epsilon >0$ and any $s = n^{\epsilon/4}, t= n^{1-\epsilon/4}$, we can partition $\iset$ into $A_1, \ldots, A_s$ and $B_1, \ldots, B_t$ in polynomial-time. Moreover, each $A_i$ is an antichain and each $B_j$ is a chain.
\end{lemma}
\sodaonly{\begin{proof}[Proof Idea]}
\fullonly{\begin{proof}[Proof idea]}
(See Section~\ref{sec:detail one} for detailed definitions and proofs.)
By Dilworth's theorem \cite{Dilworth,Fulkerson-Dilworth}, the minimum chain decomposition equals to the maximum antichain size. We will use the fact that both minimum chain decomposition and maximum-size antichain can be computed in polynomial time as follows:
As long as the maximum-size antichain is bigger than $n^{\epsilon/4}$, we repeatedly extract such an antichain out of the input; otherwise, we would have the decomposition into at most $n^{\epsilon/4}$ chains, so we stop.
\end{proof}

By Observation~\ref{observation: conseration preserving decomposition}, the collection
$\{(\cset, A_1), \ldots, (\cset, A_s), (\cset, B_1), \ldots, (\cset, B_t)\}$
 is a consideration-preserving decomposition of $(\cset, \iset)$. It follows by Lemma~\ref{lemma: UDP decomposition} that
$\sum_{i=1}^s\opt(\cset, A_i)+\sum_{j=1}^t \opt(\cset, B_j) \geq \opt(\cset, \iset).$
Further, observe that if there exists $j$ such that $\opt(\cset, B_j)\geq\opt(\cset, \iset)/(2n^{1-\epsilon/4})$, then we would be done: the $d$-{\sf UUDP-MIN} problem $\pset(\cset, B_j)$ can be seen as a $1$-{\sf UUDP-MIN} problem (since $B_j$ is a chain) and hence can be solved optimally! (See  Lemma~\ref{lem:chain to one dim} for detailed analysis)
 Otherwise $\opt(\cset, B_j)\leq \opt(\cset, \iset)/ (2n^{1-\epsilon/4})$ for every $j$. Therefore
$\sum_{j=1}^t \opt(\cset, B_j) \leq n^{1-\varepsilon/4}\cdot\opt(\cset,\iset)/ (2n^{1-\varepsilon/4})< \opt(\cset,\iset)/2.$
 If this is not the case then we know that there must be an antichain $A_i$ such that
$\opt(\cset, A_i)\geq \opt(\cset,\iset)/2n^{\epsilon/4}\,.$
%


\begin{wrapfigure}{r}{.5\textwidth}
  \vspace{-.9cm}
  \center\includegraphics[width=\linewidth]{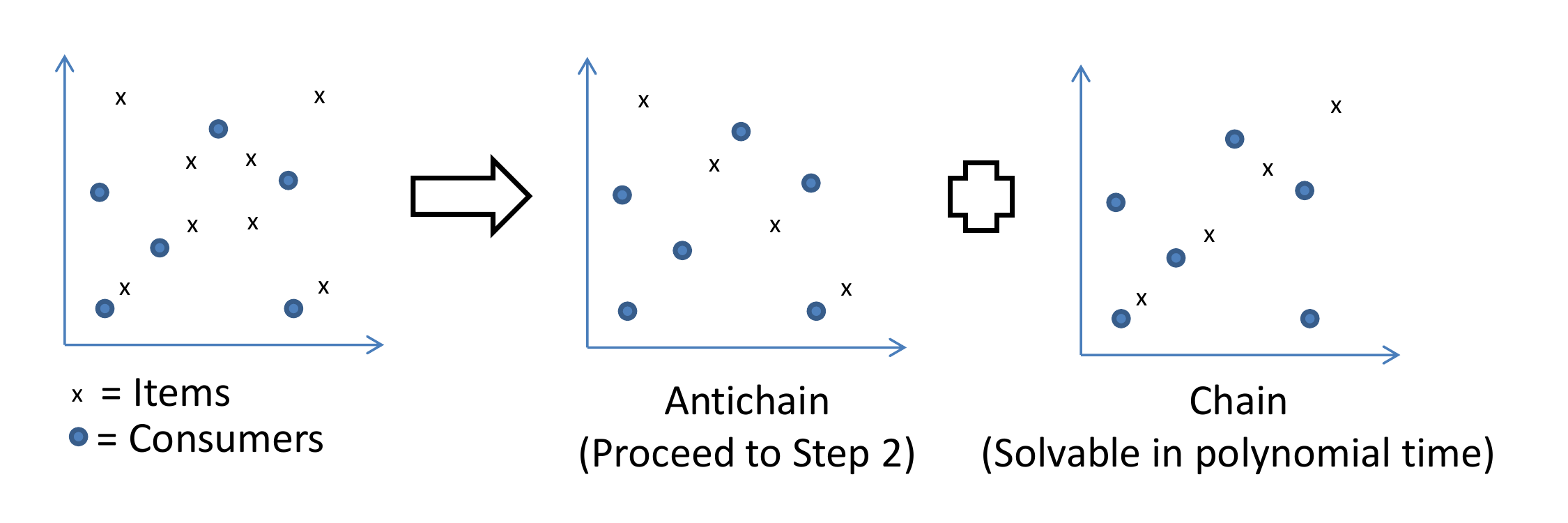}\\
  \vspace{-.3cm}\caption{Example of Step 1}\label{fig:step1}\vspace{-.2cm}
\end{wrapfigure}

\paragraph{\underline{Step 2:} Dealing with small consideration sets}
For simplicity, let us assume that we know $i$ such that $\opt(\cset, A_i)\geq \opt(\cset, \iset)/(2n^{\epsilon/4})$. Now we focus on collecting revenue from the subproblem $\pset(\cset, A_i)$. Let $\cset_1 \subseteq \cset$ be the set of consumers who are interested in at most $n^{1-2\epsilon/4}$ items in $A_i$, i.e. $\cset_1=\left\{\Consumer\in \cset: |S_\Consumer\cap A_{i}|\leq n^{1-2\epsilon/4}\right\}$, and define $\cset_2=\cset\setminus \cset_1$.
Since $\{(\cset_1, A_i), (\cset_2, A_i)\}$ is a consideration-preserving decomposition of $(\cset, A_i)$, we have
$\opt(\cset_1, A_i)+\opt(\cset_2, A_i) \geq \opt(\cset, A_i) \geq \frac{\opt(\cset, \iset)}{2n^{\epsilon/4}}.$
%
%
Now we need an algorithm of \cite{BalcanB07,BriestK11}. Balcan and Blum give an approximation algorithm for {\sf SMP} whose approximation guarantee depends on the sizes of consideration sets. Briest and Krysta, by using a slight modification of this algorithm, give an approximation algorithm with the same guarantee for {\sf UDP-MIN}. Their result, stated in terms of {\sf UUDP-MIN}, is summarized in the following theorem. (For completeness, we provide the proof in Appendix~\ref{sec:balcanblum for udp min}.)

\begin{theorem}\label{theorem:BalcanBlum}\cite{BalcanB07,BriestK11}
Given a {\sf UUDP-MIN} instance $(\cset, \iset, \set{S_{\Consumer}}_{\Consumer \in \cset})$, there is a deterministic $O(k)$-approximation algorithm of {\sf UUDP-MIN}, where  $k: = \max_{\Consumer \in \cset} \size{S_{\Consumer}}$.
\end{theorem}


We remark that we extend this technique to deal with any pricing problem with subadditive revenue in the full version of this paper.

If $\opt(\cset_1, A_i)\geq \opt(\cset, \iset)/(4n^{\epsilon/4})$, then we could invoke the algorithm in Theorem~\ref{theorem:BalcanBlum} on $(\cset_1, A_i)$ to get a solution with approximation ratio
$O\left(\max_{\Consumer \in \cset_1} \size{S_{\Consumer}\cap A_i}\right)=O(n^{1-2\epsilon/4}).$
 This yields a solution that gives a desired revenue of
$\Omega\left(\opt(\cset_1, A_i)/n^{1-2\epsilon/4}\right)=\Omega\left(\opt(\cset, \iset)/n^{1-\epsilon/4}\right)\,.$
 Otherwise we have $\opt(\cset_1, A_i)< \opt(\cset, \iset)/4n^{\epsilon/4}$. Then
$\opt(\cset_2, A_i)=\Omega\left(\opt(\cset, \iset)/n^{\epsilon/4}\right).$
We will deal with this case in the next steps.


\paragraph{\underline{Step 3:} Partitioning consumers using a small hitting set}
%
First, we apply the epsilon net theorem~\cite{Chazelle-book,Epsilon-net} to derive the following lemma.
\begin{lemma}\label{lem:hitting set}
We can find a set $H\subseteq A_i$ of size $\tilde O(n^{2\epsilon})$ in randomized polynomial time such that for any $\Consumer\in \cset_2$, there exists $\Item\in H$ such that $\Item\geq \Consumer$.
\end{lemma}
%
%
%
\begin{proof}
%
\danupon{We should consider providing a bit more detail (``a primer'') in the Appendix.}
The instance $(\cC_2, A_i)$ defines a set system $\{S_\Consumer\}_{\Consumer\in \cC_2}$ over $A_i$, where $S_\Consumer=\{\Item\in A_i\mid \Item\geq \Consumer\}$. We note that each set $S_{\Consumer}$ has {\em descriptive complexity} at most $d$,  i.e. set $S_{\Consumer}$ can be described by $d$ linear inequalities of the form $S_{\Consumer} = \bigcap_{d'=1}^d \set{\Item \in \iset: \Item[d'] \geq \Consumer[d']}$. In this case, this set system has VC dimension $O(d)$, c.f. \cite{Sharir-book}. More specifically, it is well known (e.g., \cite{AronovES10}) that any collection of $d$-dimensional axis-parallel boxes has VC dimension $O(d)$. We will not formally define VC-dimension here. The following theorem is all we need.

\begin{theorem}(\cite{Chazelle-book,Epsilon-net}; Epsilon net theorem)
Let $\xset$ be a set system of VC-dimension at most $d'$ over $N$. Then for any $\delta \in (0,1)$, we can find a set $H \subseteq N$ with $\size{H} = O(\frac{d'}{\delta}\log \frac{d'}{\delta})$ in randomized polynomial time such that, for all $X_i \in \xset$ with $\size{X_i} \geq \delta \size{N}$, it holds that  $H \cap X_i \neq \emptyset$.
\end{theorem}
Using the theorem with $\delta = n^{-2\epsilon/4}$, we can find a set $H \subseteq A_i$ of size at most $\tilde{O}(n^{2\epsilon/4})$, and since we have $|S_{\Consumer} \cap A_i| \geq \delta n$ for all $\Consumer \in \cset_2$, we are guaranteed that $H \cap S_{\Consumer} \neq \emptyset$ for all $\Consumer \in \cset_2$.
\end{proof}

We call $H$ a {\em hitting set} of $\cC_2$ since $H$ intersects $S_\Consumer$ for all $\Consumer\in \cC_2$. We use $H$ to decompose $(\cC_2, A_i)$ into a small number of subproblems and show in Step 4 that each of these problems can be viewed as a $(d-1)$-{\sf UUDP-MIN} problem.

For each $\Item\in H$, let $\cset_\Item=\{\Consumer\in \cset_2\mid \Item \in S_\Consumer\}$, i.e., $\cset_\Item$ consists of all consumers in $\cset_2$ that consider item $\Item$. Observe that $\bigcup_{\Item \in H} \cset_{\Item} = \cset_2$, and therefore by Lemma~\ref{lemma: UDP decomposition}, we have
$\sum_{\Item \in H} \opt(\cset_\Item, A_i)\geq \opt(\cset_2, A_i)
\geq \Omega\left(\opt(\cset, \iset)/n^{\epsilon/4}\right).$
%
Since $\size{H}=O(n^{2\epsilon/4})$, there exists $\Item^*\in H$ such that
\begin{align*}
\opt(\cset_{\Item^*}, A_i) &=\tilde\Omega\left(\opt(\cC, \cI)\cdot n^{-\epsilon/4}/|H|\right)
=\tilde\Omega\left(\opt(\cC, \cI)/n^{3\epsilon/4}\right).
\end{align*}
Now we, again, assume that we know $\Item^*$ and turn our focus to the subproblem $\pset(\cset_{\Item^*}, A_i)$.


\paragraph{\underline{Step 4:} Reducing the dimension} We have now reached the most crucial step.
%
%
We will (crucially) rely on the fact that all consumers in $\cset_{\Item^*}$ consider item $\Item^*$, and that $A_i$ is an antichain.
%
%
For each $j \leq d$,  define $A_i^j$ as the set of items in $A_i$ that are at least as good as $\Item^*$ in the $j$-th coordinate, i.e., $A_i^j=\{\Item\in A_i \mid \Item[j]\geq \Item^*[j]\}$. See Fig.~\ref{fig:step4_1} for an example in the case of $2$-{\sf UUDP-MIN}.

%
%
\begin{lemma}\label{lem:union of d dimensions}
$A_i = \bigcup_{j=1}^d A_i^j$.
\end{lemma}
%

\begin{figure}
\begin{center}
\subfigure[]{
\includegraphics[height=0.12\textheight, clip=true, trim=1.5cm .3cm 1cm .3cm]{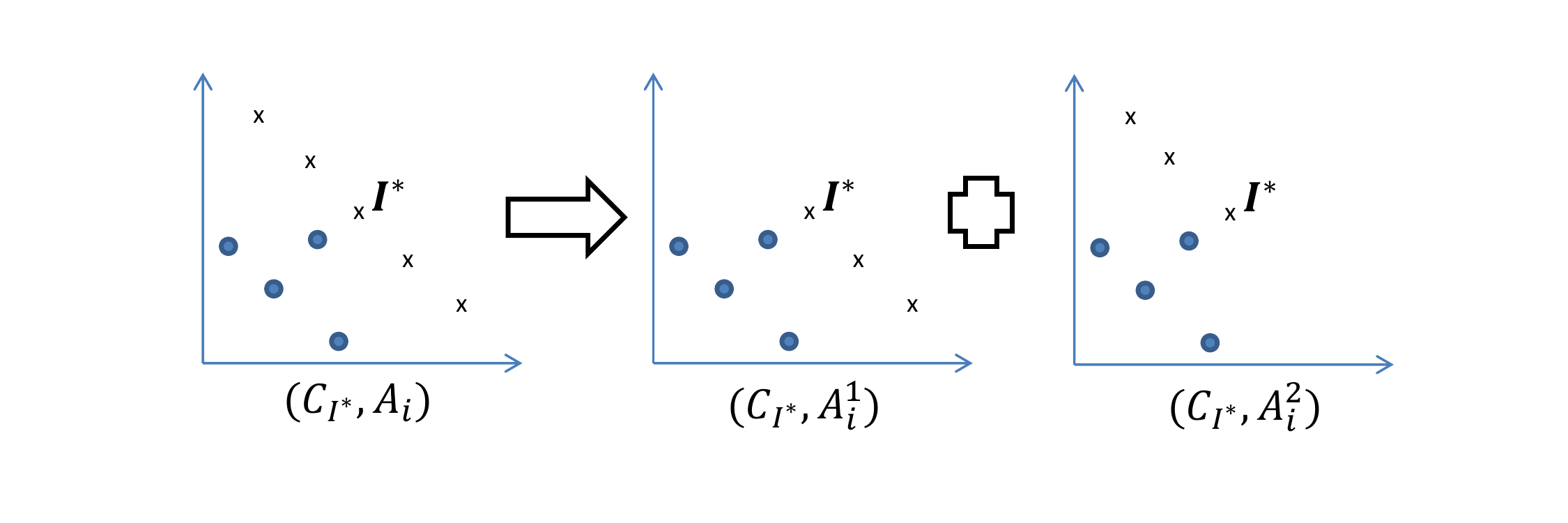}
\label{fig:step4_1}
}
\subfigure[]{
\includegraphics[height=0.12\textheight, clip=true, trim=3cm .75cm 7cm 1cm]{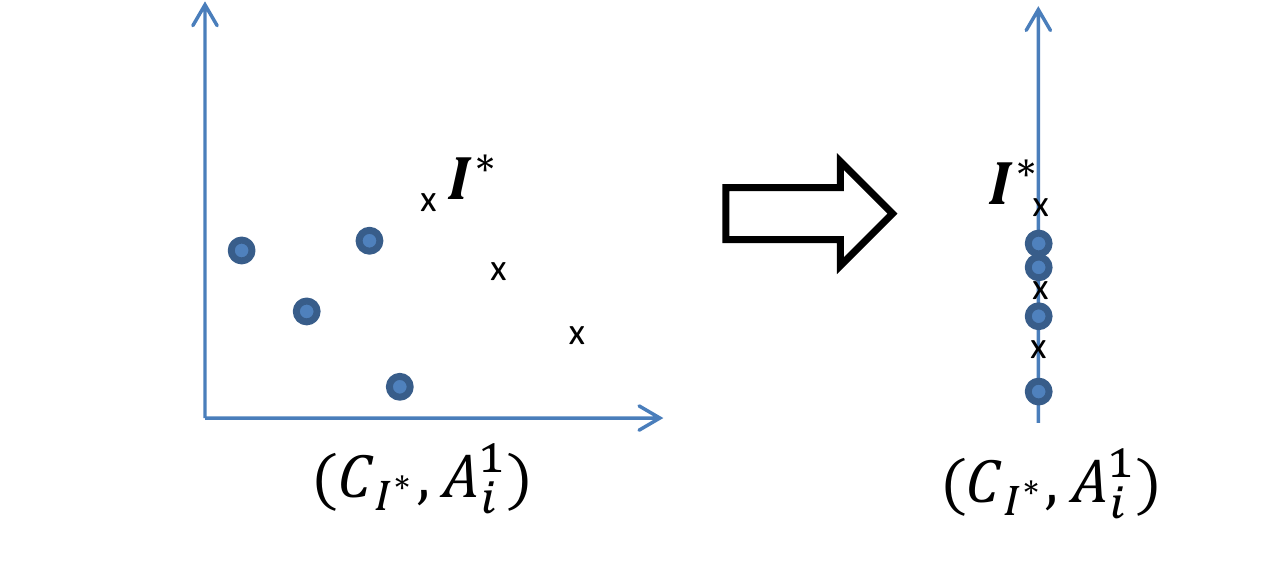}
\label{fig:step4}
}
\end{center}
\vspace{-.5cm}
\caption{\subref{fig:step4_1} Example of Step 4. \subref{fig:step4} Example of Step 4 when we view the instance $\cC_{\Item^*}, A_i^j$ as a $(d-1)$-{\sf UUDP-MIN} instance.}
\vspace{-.5cm}
\end{figure}

This lemma holds simply because $A_i$ is an antichain (in any antichain, no item can completely dominate the others, so at least one coordinate of any $\Item\in \cI_{\Item^*}$ has to be at least as good as $\Item^*$; see detailed proof in Appendix~\ref{proof:union of d dimensions}). Then $\{(\cC_{\Item^*}, A_i^1), \ldots, (\cC_{\Item^*}, A_i^d)\}$ is a consideration-preserving decomposition of $(\cC_{\Item^*}, A_i)$ and thus there exists $j$ such that
$\opt(\cC_{\Item^*}, A_i^j) \geq \opt(\cset_{\Item^*}, A_i)/d = \tilde\Omega_d(\opt(\cC, \cI)/n^{3\epsilon/4}).$
%
%
Observe that, for all $\Consumer\in \cset_{\Item^*}$ and $\Item \in A_i^j$, $\Consumer[j]\leq \Item^*[j]\leq \Item[j]$. This implies that we can ignore the $j$-th coordinate when we solve $\pset(\cset_{\Item^*}, A_i^j)$. (In particular, for any $\Consumer\in\cC_{\Item^*}$, the consideration set $S_\Consumer=\left\{\Item\geq \Consumer\mid \Item\in A_i^j\right\}$ remains the same even when we drop the $j$-th coordinate of all points.) In other words, the problem can be viewed as a $(d-1)$-{\sf UUDP-MIN} problem (see Fig.~\ref{fig:step4} for an idea).
We defer the formal statement and proof of this claim to Section~\ref{sec:detail four}.
Finally, we can invoke the $\tilde O_d(n^{1-\epsilon})$-approximation algorithm for $(d-1)$-{\sf UUDP-MIN} to collect the revenue of
%
$\tilde\Omega_d\left(\opt(\cset, \iset) n^{-3\epsilon/4}/n^{1-\epsilon}\right)$ $=\tilde\Omega_d\left(\opt(\cset, \iset)/n^{1-\epsilon/4}\right).$ Therefore we obtain an approximation ratio of $\tilde O_d(n^{1-\epsilon/4})$ in all cases. Algorithm~\ref{algo:udp} summaries our algorithm for solving $d$-{\sf UUDP-MIN}.

\begin{algorithm}
\caption{\footnotesize {\sf UUDP-MIN-APPROX}($d$)}\label{algo:udp}
\begin{algorithmic}[1]
\footnotesize
\IF{$d=1$}

\STATE Solve the problem $\cP(\cC, \cI)$ optimally using an algorithm for $1$-{\sf UUDP-MIN} (cf. Appendix~\ref{sec:one udp min})

\ELSE

\STATE Partition $\cI$ into antichains $A_1, \ldots, A_s$ and chains $B_1, \ldots, B_t$ where $s\leq n^{\epsilon/4}$ and $t\leq n^{1-\epsilon/4}$ as in Step 1.

\STATE We claim that the problems $\cP(\cC, B_1), \ldots, \cP(\cC, B_t)$ are equivalent to $1$-{\sf UUDP-MIN} problems (cf. Section~\ref{sec:detail one}). Solve them optimally using an algorithm for $1$-{\sf UUDP-MIN} (cf. Appendix~\ref{sec:one udp min}).

\FOR{$i=1, \ldots, s$}

\STATE Partition $\cC$ into $\cC_1$ and $\cC_2$ as in Step 2. Find an $O(\max_{\Consumer \in \cset_1} \size{S_{\Consumer}\cap A_i})$=$O(n^{1-2\epsilon/4})$ approximate solution of problem $\cP(\cC_1, A_i)$.


\STATE Find a hitting set $H$ of $(\cC_2, A_i)$ as in Step 3

\FOR{each $\Item\in H$}

\STATE Define $\cC_\Item$ as in Step 3

\STATE Define $A_i^1, \ldots, A_i^d$ as in Step 4 

\STATE Solve problem $\cP(\cC_\Item, A_i^1), \ldots, \cP(\cC_\Item, A_i^d)$ using an $O(n^{1-\epsilon})$-approximation algorithm for $(d-1)$-{\sf UUDP-MIN}

\ENDFOR

\ENDFOR

\ENDIF

\RETURN the solution with highest revenue among the solutions of all solved problems

\end{algorithmic}
\end{algorithm}

\subsection{Consideration-preserving Embedding}\label{sec:detail four}

To formally discuss the reduction of dimensions, we introduce the notion of consideration-preserving embedding. For any $d$, let $(\cC, \cI)$ be any instance of $d$-{\sf UUDP-MIN}. For any $d'$, consider one-to-one functions $f$ and $g$ that map points in $\R^d$ to the ones in $\R^{d'}$. We say that $(f, g)$ is a {\em consideration-preserving embedding} if, for any item $\Item \in \iset$ and consumer $\Consumer \in \cset$, we have that $\Item \geq \Consumer$ if and only if $g(\Item) \geq f(\Consumer)$. That is, the fact that consumer $\Consumer$ is considering or not considering item $\Item$ must be preserved in $f(\Consumer)$ and $g(\Item)$.

%

Given a consideration-preserving embedding $(f,g)$, we can naturally define a $d'$-{\sf UUDP-MIN} problem $\pset(f(\cset), g(\iset))$ where $f(\cset)=\{f(\Consumer)\mid \Consumer\in \cset\}$, $g(\iset)=\{g(\Item)\mid \Item \in \iset\}$ \danupon{Should we say this? ``(note that we allow the sets to contain identical points)''} and the budget $B_{f(\Consumer)}$ is $B_\Consumer$ for any $\Consumer\in\cset$.

Observe that, although  $(\cC, \cI)$ and $(f(\cC), g(\cI))$ correspond to points on different spaces, they represent the same pricing problem (i.e., the consumers' consideration sets and budgets are exactly the same). Thus, we sometimes say that $(\cC, \cI)$ and $(f(\cC), g(\cI))$ are {\em equivalent}. The following observation follows trivially.

\begin{observation}
\label{observation: consideration preserving embedding}
For any instance $(\cset,\iset)$, let $(f,g)$ be a consideration-preserving embedding of $(\cset, \iset)$ into $\R^{d'}$. Then we have that
$\opt(\cset, \iset) = \opt(f(\cset), g(\iset)).$
Moreover, if $f$ and $g$ are polynomial-time computable then a solution for $\cP(f(\cset), g(\iset))$ can be efficiently transformed into one for $\cP(\cset, \iset)$  that gives the same revenue.
\end{observation}

The transformation in the above lemma is trivial: For any price function $p$ for $(f(\cset), g(\iset))$, we simply price item $\Item\in \iset$ to $p(g(\Item))$. Observe that we will receive the same revenue from both problems using this pricing strategy.

In Step 1, we claimed that when the items form a chain, our instance would be equivalent to $1$-{\sf UUDP-MIN}. Now we prove this fact formally below.

\begin{lemma}\label{lem:chain to one dim}
Let $(\cset, \iset)$ be a $d$-{\sf UUDP-MIN} instance where $(\iset, \leq)$ is a chain. Then $(\cC, \cI)$ is equivalent to a $1$-{\sf UUDP-MIN} instance. Moreover, the corresponding consideration-preserving embedding $(f, g)$ can be computed in polynomial time.
\end{lemma}
\begin{proof}
Order items in $\cI$ by $\Item_1\leq \Item_2\leq \ldots$. Now map each item into a one-dimensional point: $g(\Item_i)=(i)$. Moreover, map each consumer according to $f(\Consumer)=g(\Item_i)$, where $i$ is the minimum number such that $\Item_i\geq \Consumer$. Observe that $(f, g)$ is a consideration-preserving embedding since $S_\Consumer=\{\Item_i, \Item_{i+1}, \ldots\}$ while $S_{f(\Consumer)}=\{g(\Item_i), g(\Item_{i+1}), \ldots\}$ for any $\Consumer\in \cset$. (Note that this embedding might create redundancy since it is possible that $f(\Consumer)=f(\Consumer')$ for some $\Consumer\neq \Consumer'$. This can be fixed easily by slightly perturbing the points.\danupon{I'm not sure if this is necessary.})
\end{proof}

In Step 4, we also claimed the dimension reduction of sub-instances $(\cset_{\Item^*}, A_i^j)$, and we now prove the claim formally. Recall that the item $\Item^*\in A_i^j$ has the property that $\Item^*\geq \Consumer$ for all $\Consumer \in \cC_{\Item^*}$ and $\Item^*[j]\leq \Item[j]$ for all $\Item\in A_i^j$.

\begin{lemma}\label{lem:dim reduction}
The instance
$(\cC_{\Item^*}, A_i^j)$ is equivalent to a $(d-1)$-{\sf UUDP-MIN} instance. Moreover, the corresponding consideration-preserving embedding $(f, g)$ can be computed in polynomial time.
\end{lemma}
\begin{proof}
Consider ``ignoring'' the $j$-th coordinate as follows. For any $\Consumer\in \cC_{\Item^*}$ and $\Item\in A_i^j$, let $f(\Consumer)=(\Consumer[1], \Consumer[2], \ldots, \Consumer[j-1], \Consumer[j+1], \ldots, \Consumer[d])$
and
$g(\Item)=(\Item[1], \Item[2], \ldots, \Item[j-1], \Item[j+1], \ldots, \Item[d]).$
Observe that for any $\Consumer\in \cC_{\Item^*}$ and $\Item\in A_i^j$,  $\Item\geq \Consumer$ trivially implies that $g(\Item)\geq f(\Consumer)$. Conversely, if $g(\Item)\geq f(\Consumer)$ then $\Item\geq \Consumer$ since $\Item[j]\geq \Item^*[j]\geq \Consumer[j]$. Thus, $(f, g)$ is a consideration-preserving embedding.
\end{proof}

\section{Hardness}\label{sec:hardness}

We provide hardness results in both  scenarios when the number of attributes $d$ is small and when $d$ is large. We sketch our results here. More details can be found in Appendix~\ref{sec:omitted_hardness}.

\paragraph{Few attributes} First we discuss the \NP-hardness of $3$-{\sf UUDP-MIN} and \APX-hardness of $4$-{\sf UUDP-MIN}. These hardness results hold even when the consumer budgets are either 1 or 2. We perform a reduction from Vertex Cover~\cite{DBLP:journals/tcs/GareyJS76,DBLP:conf/ciac/AlimontiK97}, essentially using the same ideas as in~\cite{GuruswamiHKKKM05}, except for the fact that we use Schnyder's result~\cite{Schnyder89,DBLP:conf/soda/Schnyder90} to ``embed'' the instance into posets of low order dimensions.



First, let us recall the reduction in \cite{GuruswamiHKKKM05}. We start
from a graph $G=(V,E)$, which is an input instance of Vertex Cover. We create two types of consumers: (i) poor
consumer ${\bf C}_e$ for each edge $e$ with budget $1$ and (ii) rich
consumer ${\bf C}_v$ for each vertex $v$ with budget $2$. The items
are $\iset = \set{{\bf I}_v: v\in V}$. Each poor consumer ${\bf C}_e$ has  a
consideration set containing two items ${\bf I}_u$ and ${\bf I}_v$
where $e=(u,v)$ and each rich consumer ${\bf C}_v$ considers only
one item ${\bf I}_v$. Using the analysis essentially the same as
\cite{GuruswamiHKKKM05}, one can show that the problem is {\sf NP}-hard
if we start from Vertex Cover on planar graphs and {\sf APX}-hard if
we start from Vertex Cover on cubic graphs.

Therefore, it only remains to map consumers and items to points in $\reals_{\geq 0}^d$ (where $d=3,4$) such that for each consumer
${\bf C}$, the set of items that pass her criteria (i.e., $\{\Item\in \iset \mid\Item[i]\geq \Consumer[i]~~\mbox{for all $1\leq i\leq d$}\}$) is exactly her consideration set.
The main idea is to first embed the problem into an {\em adjacency poset} of the input graph. Then, we  invoke Schnyder's theorem~\cite{Schnyder89,DBLP:conf/soda/Schnyder90} to again embed this poset into a Euclidean space.

An adjacency poset of a graph can be constructed as follows. First we construct a $2$-layer poset with minimal elements in the first layer and maximal elements in the second layer. For each edge $e \in E$, we have a minimal element in the poset corresponding to $e$ (for convenience, we also denote the poset element by $e$). For each vertex $v \in V$, we have a maximal poset element corresponding to $v$. There is a relation $e\preceq v$ if and only if vertex $v$ is an endpoint of $e$. \fullonly{This is called an {\em adjacency poset} of the graph.}


The last task is to ``embed'' poset elements into points in the Euclidean space in such a way that, for any poset elements $e_1$ and $e_2$, $e_1\preceq e_2$ if and only if $q_{e_1}[i]\geq q_{e_2}[i]$ for all $i$ where $q_{e_1}$ and $q_{e_2}$ are points that $e_1$ and $e_2$ are mapped to, respectively. If we can do this, we would be done, simply by defining the coordinates of each consumer $\Consumer_e$ to be $q_{e}$, and the coordinates of each consumer $\Consumer_v$ to be $q_v$. Similarly, we define the coordinates of each item $\Item_v$ as $q_v$. In order to obtain such an embedding, we use part of Schnyder's theorem~\cite{Schnyder89} which states that any planar graph has an adjacency poset of dimension three, and any $4$-colorable graph (including cubic graphs) has an adjacency poset of dimension four. Moreover, embedding these graphs into Euclidean spaces can be done in polynomial time~\cite{DBLP:conf/soda/Schnyder90}.

%
%

Finally we note that $2$-{\sf SMP} is strongly $\NP$-hard and $4$-{\sf SMP} is $\APX$-hard. The proof follows from the fact that these problems generalize Highway pricing and graph vertex pricing on bipartite graphs, respectively, and can be found in \sodaonly{the full version}\fullonly{Appendix~\ref{sec:omitted_hardness}}.\danupon{TO DO: Write the proof for $4$-{\sf SMP}}


%

\paragraph{Many attributes}
We establish a connection between the {\sf UUDP-MIN} with bounded-size consideration sets and our problem. This connection immediately implies hardness results for $d$-{\sf UUDP-MIN} when $d$ is at least poly-logarithmic in $n$. Our main result in this section is the following:

\begin{theorem}(Informal)
\label{theorem: higher dimension}
Let $A =(\cset, \iset, \{S_\Consumer\}_{ \Consumer \in \cset} )$ be an instance of {\sf UUDP-MIN} where $B=\max_{\Consumer \in \cset} \size{S_\Consumer}$. We can (with high probability of success) create an instance $A'=(\cset', \iset')$ of $d$-{\sf UUDP-MIN}, where $d= O(B^2 \log n)$, that is ``equivalent'' to $A$.
\end{theorem}

In other words, the above theorem shows that any {\sf UUDP-MIN} instance with consideration sets of size bounded by $B$, can be realized by a $d$-{\sf UUDP-MIN} instance for $d= O(B^2 \log n)$.
\sodaonly{
Combining this with the result in \cite{ChalermsookPricing}, we have a hardness of $\Omega(d^{1/4-\epsilon)}$ for any $\epsilon>0$. 
}
\fullonly{
This implies that $d$-{\sf UUDP-MIN} is at least as hard
as the original problem when consideration sets have size at most $B$.
When
$B$ is at least logarithmic in the number of items, our
reduction yields the following corollary, assuming the hardness of the balanced bipartite independent set problem in constant degree graphs or refuting random 3{\sf CNF} formulas \cite{BriestThesis}. \danupon{I shortened the last sentence.}
\begin{corollary}
There is a constant $\epsilon$ such that for every $d \geq \log^2
n$, it is hard to approximate $d$-{\sf UUDP-MIN} to
within a factor of $d^{\epsilon}$.
\end{corollary}
}

We remark that our reduction here in fact works independently of the decision model, so this result works for {\sf SMP} and {\sf UDP-Util} as well.

\section{Open Problems}\label{sec:conclusion}



Several interesting problems are open. The most important problem is whether we can obtain better approximation factors for $d$-{\sf UUDP-MIN} and $d$-{\sf SMP}.
%
%
We tend to believe that there is an $f(d)$-approximation algorithm for $d$-{\sf UUDP-MIN} and $d$-{\sf SMP} where $f(d)$ is a function that depends on $d$ only. However,
it seems to be a very challenging task to obtain approximation ratio like $\log^{O(d)} n$ or $O_d(\log^{1-\epsilon(d)} m)$, for some constant $\epsilon(d)>0$ depending on $d$.

\fullonly{
\begin{openproblem}
\label{conjecture: main one}
For any integer $d >0$, are there $\log^{O(d)} n$-approximation algorithms for $d$-{\sf UUDP-MIN} and $d$-{\sf SMP}?
\end{openproblem}
\begin{openproblem}
\label{conjecture: main two}
For any integer $d >0$, are there $O_d(\log^{1-\epsilon_d} m)$-approximation algorithms for $d$-{\sf UUDP-MIN} and $d$-{\sf SMP}?
\end{openproblem}
}

One promising direction in attacking the above problems is to improve Theorem~\ref{thm: dimension reduction for UDP}, e.g., getting $O_d(\rho\cdot\mathrm{polylog}(n))$ for $d$-{\sf UUDP-MIN} using a $\rho$-approximation algorithm of $(d-1)$-{\sf UUDP-MIN} as a blackbox.
\fullonly{
\begin{openproblem}
For any constant $d >0$, given a $\rho$-approximation algorithm for $(d-1)$-{\sf UUDP-MIN} (and $(d-1)$-{\sf SMP}), is it possible to get an $O_d(\rho\cdot\polylog n)$ approximation algorithm for $d$-{\sf UUDP-MIN} (and $d$-{\sf SMP})?
\end{openproblem}
}
A positive resolution to this problem would imply $(\log^{O(d)} n)$-approximation algorithm for $d$-{\sf UUDP-MIN}. We believe that, even resolving this problem would require some new insights on geometric and poset structures.




There are two special cases that can be thought of as barriers in dealing with standard versions of {\sf SMP} and {\sf UUDP-MIN}, and we believe that these two special cases serve as good starting points in attacking our problems. The first problem is the geometric version of the Maximum Expanding Subsequence ({\sc Mes}) problem which is the key problem to show the hardness of {\sf UUDP-MIN}~\cite{BriestK11}. The second problem is the Unique Coverage problem \cite{DemaineFHS08} when the sets have constant VC-dimension. Another interesting problem is to obtain {\sf PTAS}s for $2$-{\sf UUDP-MIN} and $2$-{\sf SMP} (e.g., by extending the techniques in \cite{GrandoniR10}).

\newpage

  \let\oldthebibliography=\thebibliography
  \let\endoldthebibliography=\endthebibliography
  \renewenvironment{thebibliography}[1]{%
    \begin{oldthebibliography}{#1}%
      \setlength{\parskip}{0ex}%
      \setlength{\itemsep}{0ex}%
  }%
  {%
    \end{oldthebibliography}%
  }
{ 
\bibliographystyle{plain}
\bibliography{multicriteria}

\begin{thebibliography}{10}

\bibitem{AggarwalFMZ04}
Gagan Aggarwal, Tom{\'a}s Feder, Rajeev Motwani, and An~Zhu.
\newblock {Algorithms for Multi-product Pricing}.
\newblock In {\em ICALP}, pages 72--83, 2004.

\bibitem{AggarwalH06}
Gagan Aggarwal and Jason~D. Hartline.
\newblock {Knapsack auctions}.
\newblock In {\em SODA}, pages 1083--1092, 2006.

\bibitem{AjwaniEGR12}
Deepak Ajwani, Khaled~M. Elbassioni, Sathish Govindarajan, and Saurabh Ray.
\newblock Conflict-free coloring for rectangle ranges using o(n .382) colors.
\newblock {\em Discrete {\&} Computational Geometry}, 48(1):39--52, 2012.
\newblock Also in SPAA'07.

\bibitem{DBLP:conf/ciac/AlimontiK97}
Paola Alimonti and Viggo Kann.
\newblock Hardness of approximating problems on cubic graphs.
\newblock In {\em CIAC}, pages 288--298, 1997.

\bibitem{AronovES10}
Boris Aronov, Esther Ezra, and Micha Sharir.
\newblock Small-size $\epsilon$-nets for axis-parallel rectangles and boxes.
\newblock {\em SIAM J. Comput.}, 39(7):3248--3282, 2010.
\newblock Also in STOC'09.

\bibitem{BalcanB07}
Maria-Florina Balcan and Avrim Blum.
\newblock {Approximation Algorithms and Online Mechanisms for Item Pricing}.
\newblock {\em Theory of Computing}, 3(1):179--195, 2007.
\newblock Also in EC'06.

\bibitem{BalcanB05}
Maria-Florina Balcan, Avrim Blum, Jason~D. Hartline, and Yishay Mansour.
\newblock {Mechanism Design via Machine Learning}.
\newblock In {\em FOCS}, pages 605--614, 2005.

\bibitem{BalcanBM08}
Maria-Florina Balcan, Avrim Blum, and Yishay Mansour.
\newblock {Item pricing for revenue maximization}.
\newblock In {\em ACM Conference on Electronic Commerce}, pages 50--59, 2008.

\bibitem{BansalCES06}
Nikhil Bansal, Amit Chakrabarti, Amir Epstein, and Baruch Schieber.
\newblock {A quasi-PTAS for unsplittable flow on line graphs}.
\newblock In {\em STOC}, pages 721--729, 2006.

\bibitem{Bettman79}
James~R. Bettman.
\newblock {\em An Information Processing Theory of Consumer Choice}.
\newblock Addison-Wesley, Reading, MA, 1979.

\bibitem{BettmanJP90}
James~R. Bettman, Eric~J. Johnson, and John~W. Payne.
\newblock A componential analysis of cognitive effort in choice.
\newblock {\em Organizational Behavior and Human Decision Processes},
  45(1):111--139, 1990.

\bibitem{BettmanPark80}
James~R. Bettman and C~Whan Park.
\newblock Effects of prior knowledge and experience and phase of the choice
  process on consumer decision processes: A protocol analysis.
\newblock {\em Journal of Consumer Research}, 7(3):234--48, 1980.

\bibitem{BlumH05}
Avrim Blum and Jason~D. Hartline.
\newblock {Near-optimal online auctions}.
\newblock In {\em SODA}, pages 1156--1163, 2005.

\bibitem{BriestK11}
Patrick Briest and Piotr Krysta.
\newblock Buying cheap is expensive: Approximability of combinatorial pricing
  problems.
\newblock {\em SIAM J. Comput.}, 40(6):1554--1586, 2011.
\newblock Also in SODA'07 and ICALP'08.

\bibitem{ChalermsookPricing}
Parinya Chalermsook, Julia Chuzhoy, Sampath Kannan, and Sanjeev Khanna.
\newblock Improved hardness results for profit maximization pricing problems
  with unlimited supply.
\newblock {\em APPROX}, 2012 (to appear).

\bibitem{ChalermsookKLN11}
Parinya Chalermsook, Shiva Kintali, Richard Lipton, and Danupon Nanongkai.
\newblock Graph pricing problem on bounded tree-width, bounded genus and
  k-partite graphs.
\newblock Manuscript, 2011.

\bibitem{Chan12}
Timothy~M. Chan.
\newblock Conflict-free coloring of points with respect to rectangles and
  approximation algorithms for discrete independent set.
\newblock In {\em Proceedings of the 2012 symposuim on Computational Geometry},
  SoCG '12, pages 293--302, New York, NY, USA, 2012. ACM.

\bibitem{Chazelle-book}
Bernard Chazelle.
\newblock {\em The discrepancy method: randomness and complexity}.
\newblock Cambridge University Press, New York, NY, USA, 2000.

\bibitem{ChristodoulouEF10}
George Christodoulou, Khaled~M. Elbassioni, and Mahmoud Fouz.
\newblock {Truthful Mechanisms for Exhibitions}.
\newblock In {\em WINE}, pages 170--181, 2010.

\bibitem{DemaineFHS08}
Erik~D. Demaine, Uriel Feige, MohammadTaghi Hajiaghayi, and Mohammad~R.
  Salavatipour.
\newblock {Combination Can Be Hard: Approximability of the Unique Coverage
  Problem}.
\newblock {\em SIAM J. Comput.}, 38(4):1464--1483, 2008.
\newblock Also in SODA'06.

\bibitem{Dilworth}
Robert~P. Dilworth.
\newblock {A Decomposition Theorem for Partially Ordered Sets}.
\newblock {\em Annals of Mathematics}, 51 (1):161--166, 1950.

\bibitem{ElbassioniRRS09}
Khaled~M. Elbassioni, Rajiv Raman, Saurabh Ray, and Ren{\'e} Sitters.
\newblock On the approximability of the maximum feasible subsystem problem with
  0/1-coefficients.
\newblock In {\em SODA}, pages 1210--1219, 2009.

\bibitem{ElbassioniSZ07}
Khaled~M. Elbassioni, Ren{\'e} Sitters, and Yan Zhang.
\newblock {A Quasi-{PTAS} for Profit-Maximizing Pricing on Line Graphs}.
\newblock In {\em ESA}, pages 451--462, 2007.

\bibitem{ErlebachL08}
Thomas Erlebach and Erik~Jan van Leeuwen.
\newblock Approximating geometric coverage problems.
\newblock In Shang-Hua Teng, editor, {\em SODA}, pages 1267--1276. SIAM, 2008.

\bibitem{EvenGLNV98}
Guy Even, Oded Goldreich, Michael Luby, Noam Nisan, and Boban Velickovic.
\newblock {Efficient approximation of product distributions}.
\newblock {\em Random Struct. Algorithms}, 13(1):1--16, 1998.
\newblock Also in {STOC}'92.

\bibitem{Fulkerson-Dilworth}
D.R. Fulkerson.
\newblock {Note on Dilworth's decomposition theorem for partially ordered
  sets}.
\newblock {\em Proceedings of the American Mathematical Society}, 7
  (4):701--702, 1956.

\bibitem{DBLP:conf/icalp/GamzuS10}
Iftah Gamzu and Danny Segev.
\newblock {A Sublogarithmic Approximation for Highway and Tollbooth Pricing}.
\newblock In {\em ICALP (1)}, pages 582--593, 2010.

\bibitem{DBLP:journals/tcs/GareyJS76}
M.~R. Garey, David~S. Johnson, and Larry~J. Stockmeyer.
\newblock Some simplified {NP}-complete graph problems.
\newblock {\em Theor. Comput. Sci.}, 1(3):237--267, 1976.

\bibitem{Gensch87}
Dennis~H. Gensch.
\newblock {A two-stage disaggregate attribute choice model}.
\newblock {\em Marketing Science}, 6(3):223--239, 1987.

\bibitem{GilbrideAllenby}
Timothy~J. Gilbride and Greg~M. Allenby.
\newblock {A Choice Model with Conjunctive, Disjunctive, and Compensatory
  Screening Rules}.
\newblock {\em Marketing Science}, 23(3):391--406, 2004.

\bibitem{GrandoniR10}
Fabrizio Grandoni and Thomas Rothvo{\ss}.
\newblock {Pricing on Paths: A {\sf PTAS} for the Highway Problem}.
\newblock In {\em SODA}, pages 675--684, 2011.

\bibitem{GuruswamiHKKKM05}
Venkatesan Guruswami, Jason~D. Hartline, Anna~R. Karlin, David Kempe, Claire
  Kenyon, and Frank McSherry.
\newblock {On profit-maximizing envy-free pricing}.
\newblock In {\em SODA}, pages 1164--1173, 2005.

\bibitem{HaublTrifts00}
Gerald Ha\"ubl and Valerie Trifts.
\newblock {Consumer Decision Making in Online Shopping Environments: The
  Effects of Interactive Decision Aids}.
\newblock {\em Marketing Science}, 19(1):4--21, 2000.

\bibitem{HauserTEBS10}
John~R. Hauser, Olivier Toubia, Theodoros Evgeniou, Rene Befurt, and Daria
  Silinskaia.
\newblock {Disjunctions of Conjunctions, Cognitive Simplicity and Consideration
  Sets}.
\newblock {\em J. Marketing Res.}, 47(3):485--496, 2010.

\bibitem{HauserW90}
John~R Hauser and Birger Wernerfelt.
\newblock An evaluation cost model of consideration sets.
\newblock {\em Journal of Consumer Research}, 16(4):393--408, 1990.

\bibitem{Epsilon-net}
David Haussler and Emo Welzl.
\newblock {epsilon-Nets and Simplex Range Queries}.
\newblock {\em Discrete {\&} Computational Geometry}, 2:127--151, 1987.

\bibitem{ItoNOOUUU12}
Takehiro Ito, Shin-Ichi Nakano, Yoshio Okamoto, Yota Otachi, Ryuhei Uehara,
  Takeaki Uno, and Yushi Uno.
\newblock A polynomial-time approximation scheme for the geometric unique
  coverage problem on unit squares.
\newblock In Fedor~V. Fomin and Petteri Kaski, editors, {\em SWAT}, volume 7357
  of {\em Lecture Notes in Computer Science}, pages 24--35. Springer, 2012.

\bibitem{JedidiK05}
Kamel Jedidi and Rajeev Kohli.
\newblock {Probabilistic Subset-Conjunctive Models for Heterogeneous
  Consumers}.
\newblock {\em J. Marketing Res.}, 42(4):483--495, 2005.

\bibitem{KhandekarKMS09}
Rohit Khandekar, Tracy Kimbrel, Konstantin Makarychev, and Maxim Sviridenko.
\newblock On hardness of pricing items for single-minded bidders.
\newblock In {\em APPROX-RANDOM}, pages 202--216, 2009.

\bibitem{KortsarzRR11}
Guy Kortsarz, Harald Racke, and Rajiv Raman.
\newblock Gap between unpward and tolbooth instances.
\newblock Manuscript, 2011.

\bibitem{LiuArora2011}
Qing Liu and Neeraj Arora.
\newblock {Efficient Choice Designs for a Consider-Then-Choose Model}.
\newblock {\em Marketing Science}, 30(2):321--338, 2011.

\bibitem{Payne82}
John~W. Payne.
\newblock {Contingent decision behavior}.
\newblock {\em Psychological Bulletin}, 92(2):382--402, 1982.

\bibitem{PayneBJ88}
John~W. Payne, James~R. Bettman, and Eric~J. Johnson.
\newblock {Adaptive strategy selection in decision making.}
\newblock {\em Journal of Experimental Psychology: Learning, Memory, and
  Cognition}, 14(3):534--552, 1988.

\bibitem{Rusmevichientong03}
Paat Rusmevichientong.
\newblock A non-parametric approach to multi-product pricing: Theory and
  application.
\newblock {\em Ph. D. thesis, Stanford University}, 2003.

\bibitem{RusmevichientongRG06}
Paat Rusmevichientong, Benjamin~Van Roy, and Peter~W. Glynn.
\newblock {A Nonparametric Approach to Multiproduct Pricing}.
\newblock {\em Operations Research}, 54(1):82--98, 2006.

\bibitem{RusmevichientongRG06-2}
Paat Rusmevichientong, Joyce~A Salisbury, Lynn~T Truss, Benjamin~Van Roy, and
  Peter~W. Glynn.
\newblock {Opportunities and challenges in using online preference data for
  vehicle pricing: A case study at General Motors}.
\newblock {\em Journal of Revenue and Pricing Management}, 5:45--61, 2006.

\bibitem{Schnyder89}
Walter Schnyder.
\newblock {Planar graphs and poset dimension}.
\newblock {\em Order}, 5:323--343, 1989.

\bibitem{DBLP:conf/soda/Schnyder90}
Walter Schnyder.
\newblock Embedding planar graphs on the grid.
\newblock In {\em SODA}, pages 138--148, 1990.

\bibitem{Sharir-book}
Micha Sharir and Pankaj~K. Agarwal.
\newblock {\em Davenport-Schinzel Sequences and Their Geometric Applications}.
\newblock Cambridge University Press, New York, NY, USA, 1995.

\bibitem{Shocker91}
Allan~D. Shocker, Moshe Ben-Akiva, Bruno Boccara, and Prakash Nedungadi.
\newblock {Consideration set influences on consumer decision-making and choice:
  Issues, models, and suggestions}.
\newblock {\em Marketing Letters}, 2(3):181--197, 1991.

\bibitem{Shugan80}
Steven~M. Shugan.
\newblock {The cost of thinking}.
\newblock {\em Journal of Consumer Research}, 7:99--111, 1980.

\end{thebibliography}
}

\newpage
\appendix
\section*{Appendix}
\section{Proof Omitted from Section~\ref{sec:d udp min}}

\subsection{Proof of Lemma~\ref{lemma: UDP decomposition}}
\label{sec: proof of UDP decomposition lemma}

Let $p^*$ be the optimal price function for $\pset(\cset, \iset)$. For each $i=1,\ldots, k$, we define $p^*_i: \iset'_i \rightarrow \R$ by
$$p^*_i(\Item) = p^*(\Item)~~\mbox{if $\Item \in \iset'_i$, and $p^*_i(\Item) = \infty$ otherwise.}$$
Let $r_i$ be the total revenue made by $p^*_i$ in $\pset(\cset'_i, \iset'_i)$. We argue below that
\begin{align}\label{eq:one}
\sum_{i=1}^k r_i \geq \opt(\cset, \iset).
\end{align}
Let $\cset^* \subseteq \cset$ be the set of consumers who make a positive payment with respect to $p^*$. For each consumer $\Consumer \in \cset^*$, denote by $\phi(\Consumer) \in \iset$ the item that consumer $\Consumer$ buys with respect to the price $p^*$.
So we can write $\opt(\cset, \iset)$ as
\begin{align}\label{eq:two}
\opt(\cset, \iset) = \sum_{\Consumer \in \cset^*} p^*(\phi(\Consumer)).
\end{align}
For each $i=1,\ldots, k$, let $\cset^*_i \subseteq \cset'_i$ be the set of consumers $\Consumer \in \cset'_i$ such that $\phi(\Consumer) \in \iset'_i$. That is, $\cset^*_i$ is a set of consumers whose item she bought in $\opt(\cset, \iset)$ is in $\iset'_i$.
Notice that
\begin{align}\label{eq:three}
r_i \geq \sum_{\Consumer \in \cset^*_i} p^*(\phi(\Consumer)).
\end{align}
Since $\set{(\cset'_i, \iset'_i)}_{i=1}^k$ is a consideration-preserving decomposition, we have that
\begin{align}\label{eq:four}
\bigcup_{i=1}^k \cset^*_i \supseteq \cset^*,
\end{align}
since for any $\Consumer \in \cset^*$, we must have $\phi(\Consumer)\in \iset_i$ for some $i$. By summing Eq.\eqref{eq:three} over all $i=1,\ldots, k$, we have
\begin{align*}
\sum_{i=1}^k r_i &\geq \sum_{i=1}^k \sum_{\Consumer\in \cset^*_i} p^*(\phi(\Consumer)) &\mbox{(by Eq.\eqref{eq:three})}\\
&\geq \sum_{\Consumer\in \cset^*} p^*(\phi(\Consumer)) &\mbox{(by Eq.\eqref{eq:four})}\\
&=\opt(\cset, \iset)&\mbox{(by Eq.\eqref{eq:two})}
\end{align*}
%
%
This proves Eq.\eqref{eq:one} and thus the first claim.

Now suppose we have a price $p': \iset_i \rightarrow \R$ that collects revenue $r'$ in $\pset(\cset'_i, \iset'_i)$. We define a function $p: \iset \rightarrow \R$ by $p(\Item) = p'(\Item)$ for $\Item \in \iset'_i$ and $p(\Item) = \infty$ otherwise. We can use $p'$ to obtain a revenue of $r'$ from $\cP(\cset, \iset)$. This proves the second claim.

\subsection{Decomposing items into small number of chains and antichains}\label{sec:detail one}

We will use the following theorem, first proved by Dilworth \cite{Dilworth}, and its polynomial computability follows from the equivalence between Dilworth's theorem and K\"onig's theorem~\cite{Fulkerson-Dilworth}.

\begin{theorem}\label{theorem:dilworth}
Let $(S, \leq)$ be a partially ordered set, and $Z$ be the maximum number of elements in any antichain of $S$. Then there is a polynomial-time algorithm that produces a partition of $S$ into $Z$ chains $S_1,\ldots, S_Z$.
\end{theorem}

We now use the theorem to prove Lemma~\ref{lemma: chain decomposition}.

\begin{proof}[of Lemma~\ref{lemma: chain decomposition}]
Initially, let $i=1$. In iteration $i$, we check if the size of maximum antichain in $\iset$ is at least $t=n^{1-\epsilon/4}$. If so, we find the maximum antichain $A_i$, update $\iset = \iset \setminus A_i$, and proceed to the next iteration; otherwise, we stop the iterations. Notice that the number of iterations is at most $s = n^{\epsilon/4}$, and when the iteration stops, the size of maximum-size antichain is at most $t\leq n^{1-\epsilon/4}$. We apply the above theorem to compute a decomposition of $\iset$ into $t$ chains, denoted by $B_1,\ldots, B_t$. This concludes the proof of Lemma~\ref{lemma: chain decomposition}.
\end{proof}

\subsection{Proof of Balcan-Blum Theorem for {\sf UUDP-MIN} (cf. Theorem~\ref{theorem:BalcanBlum})}\label{sec:balcanblum for udp min}

%
We first explain a randomized algorithm, and then we discuss how to derandomize it. This part is essentially the same as \cite{BalcanB07,BriestK11}. First, we randomly construct a set $\iset^* \subseteq \iset'$ where each item ${\bf I}$ is independently added to $\iset^*$ with probability $1/k$ (recall that $k=\max_{\Consumer\in \cset} |S_\Consumer|$). Then let $\cset^*$ be a set of consumer $\Consumer$ such that $\size{S_\Consumer\cap \iset^*}=1$ 
(i.e. consumers who care about exactly one item in $\iset^*$). We show that the problem $\cP(\cset^*, \iset^*)$ has expected revenue at least $\Omega(\opt(\cset, \iset)/k)$.

Let $p$ be the optimal price function for $(\cset, \iset)$ and $\phi: \cset \rightarrow \iset\cup \{\perp\}$ be a function that maps each consumer to the item she buys with respect to $p$ (let $\phi(\Consumer)=\perp$ if consumer $\Consumer$ buys nothing and $p(\perp)=0$). Therefore, we have that $\opt(\cset, \iset) = \sum_{\Consumer} p(\phi(\Consumer))$. We denote by $p^*$ the price function $p$ restricted to $\iset^*$. For each $\Consumer$, if $\Consumer \in \cset^*$ and $\phi(\Consumer) \in \iset^*$, the revenue created by $p^*$ in $(\cset^*, \iset^*)$ would be at least $p(\phi(\Consumer))$. Therefore,
$$\expect{}{\opt(\cset^*, \iset^*)} \geq \sum_{\Consumer \in \cset} \Pr[\mbox{$\phi(\Consumer)\in \iset^*$ and $\Consumer\in \cset^*$}] \times p(\phi(\Consumer))\,.$$
Notice that, for any $\Consumer\in\cset$ and $\Item\in S_{\Consumer }$,
\[\Pr[\mbox{$\Item\in \iset^*$ and $\Consumer\in \cset^*$}] \geq \frac{1}{k}\left(1-\frac{1}{k}\right)^{k-1}\geq \frac{1}{k\mathrm{e}},\]
which implies that $\expect{}{\opt(\cset^*, \iset^*)} \geq \frac{1}{k \mathrm{e}}\cdot \opt(\cset, \iset)$.

{\bf Derandomization:} First, note that we can assume that $k=O(\log m+\log n)$. Otherwise, we can use the result of \cite{AggarwalFMZ04,GuruswamiHKKKM05,BalcanBM08} (see \cite[Section 4]{BalcanBM08} for the result in a general setting) to obtain $O(\log m+ \log n)$ approximation algorithm for {\sf UUDP-MIN}, which will also be $O(k)$-approximation.

Now, assuming that $k=O(\log m+\log n)$, we follow the argument of Balcan and Blum~\cite{BalcanB07}. In particular, we observe that we need only $k$-wise independence among the events of the form ``$\Item\in \iset^*$ and $\Consumer\in \cset^*$'', for any $\Item$ and $\Consumer$,  in order to get the above expectation result. In this case, we can use the tools from Even et al \cite{EvenGLNV98} to derandomize the above algorithm while blowing up the running time by a factor of $2^{O(k)}=\poly(m, n)$. For more details, we refer the readers to \cite{BalcanB07}.

%
%
%
%

\subsection{Proof of Lemma \ref{lem:union of d dimensions}}\label{proof:union of d dimensions}
Recall that each $A_i$ is an antichain, i.e., for any distinct $\Item_1, \Item_2\in A_i$, there exists $1\leq d_1, d_2\leq d$ such that $\Item_1[d_1]<\Item_2[d_1]$ and $\Item_1[d_2]>\Item_2[d_2]$. In particular, if $\Item_1=\Item^*$, then we have that for any $\Item\in A_i$, there exists coordinate $j$ such that $\Item[j]\geq \Item^*$. This means that $\Item\in A_i^j$. The lemma follows.

\subsection{Polynomial-Time Algorithm for $1$-{\sf UUDP-MIN}}\label{sec:one udp min}

We provide a polynomial-time algorithm for solving $1$-{\sf UUDP-MIN}. Let $\Item_1,\ldots, \Item_n$ be a sequence of items ordered non-increasingly by their coordinates. We can assume without loss of generality that their coordinates are different (by slightly perturbing their values), and we say that consumer $\Consumer$ is at {\em level $j$} if her coordinate lies between $\Item_{j-1}$ and $\Item_j$. Notice that, for any consumer $\Consumer$ at level $j$, we have $S_{\Consumer} = \set{\Item_1,\ldots, \Item_j}$.

\begin{claim}
Let $p^*$ be an optimal price. Then we can assume that $p^*(\Item_1) \geq p^*(\Item_2) \geq \ldots \geq p^*(\Item_n)$.
\end{claim}
\begin{proof}
Suppose that $p^*(\Item_i) < p^*(\Item_j)$ for some $i <j$. Recall that $\Item_i \geq \Item_j$, so for each consumer $\Consumer$ such that $\Consumer \leq \Item_j$, we know that $\Consumer$ does not buy item $\Item_j$ with respect to this solution. Thus, we can reduce $p^*(\Item_i)$ slightly, while maintaining the same revenue.
\end{proof}

The claim will ensure that consumers at level $j$ only buy item $\Item_j$ but not any other items in $\set{\Item_1,\ldots, \Item_{j-1}}$, and this allows us to solve the problem by dynamic programming. For each $j=1,\ldots, n$, for each price $P \in \R$ we have a table entry $T[j,P]$ that keeps the maximum revenue achievable from consumers at levels $1,\ldots, j$ and items $\set{\Item_1,\ldots, \Item_j}$ where the price of $\Item_j$ is set to $P$. Notice that it is easy to compute the profit from consumers at level $j$ if we know $p(\Item_j) = P$. Denote such value by $\gamma$. Then we have that $T[j,P] = \gamma + \max_{P' \geq P} T[j-1,P']$. Finally, we note that there are at most $|\cset|$ possibilities of prices $P$ because one can assume without loss of generality that, for {\sf UUDP-MIN}, the prices always belong to $\set{B_{\Consumer}}_{\Consumer \in \cset}$.

\section{{\sf QPTAS} for $2$-{\sf UUDP-MIN}}
\label{sec: bicriteria}\label{sec: qptas 2 udp}


%
%
%

We note that we will write $O(\log m)$ instead of $O(\log n+\log m)$ since we assume that $n\leq m$ in this paper. (Otherwise, we already have approximation ratio of $O(\log m)=O(\log n)$.)

\sodaonly{

We explain the main idea first. The intuition can be realized by solving the following simple case: Assume for now that we have $\Theta(n^2)$ items, which form a set $\set{(2i-1,2j-1): 1 \leq i,j \leq n}$. In this case it is
possible to have two different consumers at the same coordinate, i.e. $\Consumer=\Consumer'$, while there is exactly one item at each point $(2i-1,2j-1)$. Assume further that each consumer has budget either $1$
or $2$. We show below how to solve this case in polynomial time.

Note that there is an optimal solution such that each item is
priced either $1$ or $2$: otherwise we could increase the price by small amount to collect more revenue.
Now, for any item point $(2i-1, 2j-1)$ and any price assignment
$p$, define\danupon{(to polish after submission) $r_p(i, j)$ may be confused with $r_\Consumer(p)$.}
\[r_p(i, j):=\min_{\substack{\Item[1]\ge 2i-1, \Item[2]\geq 2j-1 \\ \Item\in\iset}}\{p(\Item)\}
\]
 to be the minimum price among the items
dominating $(2i-1, 2j-1)$. This quantity immediately tells us how much
revenue we will get from consumers at point $(2i-2, 2j-2)$: each consumer
will buy an item at price $r_p(i, j)$ if and only if she has budget
at least $r_p(i, j)$.

By the definition of $r_p$, we know that for any fixed value $j$, $r_p(i,j)$ is non-decreasing in terms of $i$.
In other words, for any pricing $p$ and integer $j$, there exists
a ``threshold'' $\gamma(p, j)$ such that $r_p(i', j)=1$ for all $i'\leq
\gamma(p, j)$ and $r_p(i',j)=2$ for all $i'> \gamma(p, x)$.
Additionally, for any $j$, $\gamma(p, j)\geq \gamma(p, j+1)$.
Using these observations, we are ready to define the dynamic programming table. The table entry $T[i,j]$ is defined to be the maximum revenue we can get among the price assignment $p$ such that $r_p(i', j)=1$ for all
$i'\leq i$ and $r_p(i', j)=2$ for all $i'>i$. The table $T$ can be
computed as follows.
\begin{align}
T[i, j]&=\max_{i'\leq i} \{T[i',j+1]+m_1(i',j) + 2m_2(i',j)\}
\label{eq:udp-min-table}
\end{align}
where $m_1(i',j)$ is the number of consumers of the form
$(2i''-2,2j-2)$ for $i''\leq i'$ with budget $1$ and $m_2(i',j)$ is the number of
consumers of the form $(2i''-2,2j-2)$ for $i'' > i'$ with budget $2$. Moreover, let $T[i,n+1]=0$ for all $i$.
The optimal solution is then $\max_i T[i,1]$.

The above discussion captures almost all the key ideas for solving the general $2$-{\sf UUDP-MIN} problem. To get a {\sf QPTAS} in the general case,
we extend these ideas in the following way.

\squishlist

\item Consider a slight generalization when there is only one item in each column and row of grid cells (cf.
Lemma~\ref{lem:perturb}) while each budget is still $1$ and $2$. In
this case, we cannot pick arbitrary value of $i'$ when we compute
$T[i, j]$ as in Eq.\eqref{eq:udp-min-table} since it might not
correspond to any pricing. Through some additional observations,
table $T$ can be computed as follows: Let $\Item_j$ be the item whose $y$-coordinate is $j$. If $i=\Item_j[1]$ then we can use Eq.\eqref{eq:udp-min-table}; otherwise, $T[i,j]=T[i,j+1]+m_1(i,j)+2m_2(i,j)$. This algorithm runs in $O(n^3)$ time.

\item When there are $q$ different budgets, say $B_1, B_2, \ldots, B_q$,
we can solve the problem in $n^{O(q)}$ time. This is done
by defining $T[i_1,\ldots, i_{q-1}, j]$ to be the maximum
revenue we can get among the price assignment $p$ such that, for all
$q': 1 \leq q' \leq q$, $r_p(i',j)=B_{q'}$ for all $i_{q'-1}<i'\leq i_{q'}$ (where we let
$i_0=-1$ and $i_q=n$).

\item Finally, we obtain a {\sf QPTAS} by ``discretizing'' the prices so that there are not many choices of item prices (cf. Lemma~\ref{lem:uudp-discretize}). This enables us to assume that the prices are in  $\Gamma=\{0, (1+\epsilon)^0, (1+\epsilon)^1, ..., (1+\epsilon)^q\}$ where $q=O(\log_{1+\epsilon} m)$, and we can get the algorithm running in time $n^{O(\size{\Gamma})}=n^{O(\log m n)}$.\danupon{I changed this slightly: I point to Lemma~\ref{lemma:discretization}.}

\squishend
}

\subsection{Preprocessing}

\fullonly{We need to do a preprocessing, which will be used in the next two sections to design {\sf QPTAS} for $2$-{\sf UUDP-MIN} and $2$-{\sf SMP}.}
The following lemma says that we can assume the input lies on the grid where each row and column of the grid contains exactly one item.

\begin{lemma}\label{lem:perturb}
We are given an instance $(\cset,\iset)$ of $2$-{\sf UUDP-MIN}\fullonly{ (or $2$-{\sf SMP})}. Then we can, in polynomial time, transform $(\cset, \iset)$ into an ``equivalent'' instance $(\cset', \iset')$ such that
\begin{itemize}
\item Each consumer ${\bf C}' \in \cset'$ has even coordinates $(2i, 2j)$ for $0 \leq i,j \leq n$.
\item Each item ${\bf I}' \in \iset'$ has odd coordinate $(2i-1, 2j-1)$ for $1 \leq i,j \leq n$.

\item For each odd number $2i-1$, $1 \leq i \leq  n$, there is exactly one item ${\bf I}' \in \iset'$ with ${\bf I}'[1]=2i-1$ and exactly one item ${\bf I}'$ with ${\bf I}'[2]=2i-1$.
\end{itemize}
\end{lemma}
\begin{proof}
We sweep the horizontal line from top to bottom, and whenever the
line meets the items ${\bf I'}_1,\ldots, {\bf I'}_z$ such that ${\bf
I}'_1[1] < {\mathbf I}'_2[1] < \ldots < {\mathbf I}'_z[1]$ with the
same $y$-coordinate $y'$, we break ties as follows. Let $\delta$ be
the vertical distance from the line to the next item point below the
line. We set the new $y$-coordinates of these items to ${\bf
I'}_j[2] = y'-(z-j)\delta/2 z$. Notice that some consumers whose
$y$-coordinates lie in $[y', y'-\delta)$ get affected by this move.
We also change the $y$-coordinates of those consumers to
$y'-\delta/2$. Then we add the horizontal grid lines between the
space of every consecutive items, while making sure that consumer
points are on the line passing $y-\delta/2$. It is easy to see that
this process preserves the consideration set of every consumer. We
repeat the above steps until the sweeping line passes the bottommost
item.

We do a similar sweep of vertical line from right to left, inserting
the grid lines along the way. In the end, each consumer lies on the
intersection of the grid lines and each item in its cell, which guarantees
that no two items appear in the same row or column of the grid.
\end{proof}

\subsection{Detail of {\sf QPTAS} for {\sf UUDP-MIN}}

First, we can make the following simple assumption.

\begin{lemma}
\label{lem:uudp-discretize}
We can assume that the prices are in the form $(1+\epsilon)^0,
(1+\epsilon)^1, ..., (1+\epsilon)^q$ or zero where
$q=O(\log_{1+\epsilon} m)$ by sacrificing $(1+\epsilon)$ in the approximation factor.
\end{lemma}

\begin{proof}
We use the following standard arguments. Consider an optimal price
$p^*$. For each item ${\bf I}_j$, if the price is non-zero, we round
down the price $p^*({\bf I}_j)$ to the nearest scale of
$(1+\epsilon)^{q'}$, so the price of each item gets decreased by at
most a factor of $(1+\epsilon)$. Consider a consumer $\Consumer$ who
bought ${\bf I}_{j}$ with price $p^*$. After the rounding, she can
still afford ${\bf I}_j$, so we can still collect at least
$(1-\epsilon)p^*({\bf I}_j)$ from $C$.\danupon{Missing: We have to
first show that we can bound the maximum budget by $O(m)$.}
\end{proof}

Now, assuming that the optimal price $p^*$ has the above structure,
we show how to solve the problem in quasi-polynomial time. First, we
reorder the items based on their $y$-coordinates in descending
order, so we have ${\bf I}_1[2] > {\bf I}_2[2] > \ldots > {\bf
I}_n[2]$. A consumer $\Consumer$ is said to belong to {\em level
$j$} if it lies  between the row\danupon{We didn't define ``row'' before.} of ${\bf I}_j$ and
that of ${\bf I}_{j+1}$, so each consumer belongs to exactly one
level. Moreover, observe that a consumer $\Consumer$ at level $j$ is
only interested in (a subset of) items in $\set{{\bf I}_1,\ldots,
{\bf I}_j}$ (since $\Item_{j'}[2]<\Consumer[2]$ for any $j'>j$).
We define a subproblem $\pset_j$ as the pricing problem with items
$\set{{\bf I}_1,\ldots, {\bf I}_j}$ and consumers at levels
$1,\ldots, j$.  We use the dynamic programming technique to solve this
problem.

\paragraph{Profiles} We will remember the profile for each
subproblem $\pset_j$. A profile $\Pi$ of $\pset_j$ consists of
$O(\log m)$ item indices $\pi_1,\ldots, \pi_q \in \set{1,\ldots,
j}$. Each value $\pi_i$ is supposed to tell us the index of the item $\Item$
of price $(1+\epsilon)^i$ with maximum value $\Item[1]$. That is, we say
that a price $p$ for $\pset_j$ is {\em consistent} with profile
$\Pi =(\pi_1,\ldots, \pi_q)$ if, for each $i$, the item ${\bf I}_{\pi_i}$ has the highest value in the first coordinate among the items with price at most $(1+\epsilon)^i$, i.e., for all
$i$,
$$\pi_i=\arg\max_{j'} \{ \Item_{j'}[1]\ |\ p(\Item_{j'})\leq (1+\epsilon)^i\}\,.$$
Since $\{ \Item_{j'}\ |\ p(\Item_{j'})\leq (1+\epsilon)^i\}\subseteq
\{ \Item_{j'}\ |\ p(\Item_{j'})\leq (1+\epsilon)^{i+1}\}$ for any
$i$,
$${\bf I}_{\pi_1}[1] \leq {\bf I}_{\pi_2}[1] \leq \ldots \leq {\bf
I}_{\pi_q}[1]\,.$$

Observe that  if two prices $p'$ and $p''$ have the
same $\pset_j$ profile, then consumers at level $j$ see no
difference between these two prices, as shown formally by the
following lemma. We say that an item $\Item_k$ is a profile item for profile $\Pi=(\pi_1,\ldots, \pi_q)$ if and only if $k = \pi_{q'}$ for some $q' \in [q]$.

\begin{lemma}
\label{lemma:reconstruction}
Let $\Pi$ be a profile for subproblem $\pset_j$, and let $p$ be any
price function for $\pset_j$ that is consistent with profile $\Pi$. Then we
can assume without loss of generality that every consumer at level $j$ only purchases
profile items.
\end{lemma}

\begin{proof}
Suppose that a consumer $\Consumer$ buys an item ${\bf I}$ in $\iset$
with $p({\bf I}) =(1+\epsilon)^{q'}$ which is not a profile item.
Then consider the profile item ${\bf I}_{\pi_{q'}}$, which satisfies ${\bf I'}[1] \geq {\bf I}[1]$, so we must have $\Item_{\pi_{q'}} \in S_{\Consumer}$. We can
therefore assume that consumer $\Consumer$ buys $\Item_{\pi_{q'}}$ instead
of ${\bf I}$.
\end{proof}

Let $\Pi=(\pi_1,\ldots, \pi_q)$ be a profile for $\pset_{j}$ and
$\Pi'= (\pi'_1,\ldots, \pi'_q)$ be a profile for $\pset_{j-1}$. We
say that $\Pi$ is {\em consistent} with $\Pi'$ if for any price
$p':\set{{\bf I}_1,\ldots, {\bf I}_{j-1}} \rightarrow \R$ that is
consistent with $\Pi'$, we can extend $p'$ to $p$ by assigning value
$p({\bf I}_j)$ such that $p$ is consistent with $\Pi$. Notice that
consistency between any two profiles can be checked in polynomial
time by trying all $q$ possibilities of prices.

We recall that we use $p^*$ to denote the optimal price.

\begin{lemma}\label{lem:udp-consistency}
There are profiles $\Pi^1,\ldots, \Pi^n$ for $\pset_1,\ldots,
\pset_n$ respectively such that for each $j\in\{1,\cdots, n-1\}$,
$\Pi^j$ is consistent with $\Pi^{j+1}$. Moreover, all such profiles
are consistent with price $p^*$.
\end{lemma}
\begin{proof}
For each subproblem $\pset_j$, we define the profile
$\Pi^j=(\pi^j_1,\ldots, \pi^j_q)$ based on the price $p^*$ (there is
only one possible profile consistent with $p^*$). It is clear that
$\Pi^j$ is always consistent with $\Pi^{j+1}$.
\end{proof}

\paragraph{Dynamic Programming Table} For each $j=1,\ldots, n$ and for each profile $\Pi$ of $\pset_j$, we use a table entry $T(j,\Pi)$ to store the maximum revenue achievable among the price function for $\pset_j$ that is consistent with the profile $\Pi$. Since there are $n^{O(\log m)}$ possibilities for the profile $\Pi$, the table size is $n^{O(\log m)}$. We now show the computation of the table. To compute $T(j,\Pi)$, we recall that given the profile $\Pi$, the revenue from consumers at level $j$ can be computed efficiently. Denote such revenue by $r_j(\Pi)$. The following equation holds:
\[T(j,\Pi) = r_j(\Pi) + \max_{\Pi' \mbox{ consistent with } \Pi} T(j-1, \Pi') \]

\paragraph{Computing the Solution} For each table entry $T(j,\Pi)$, we can keep track of the profile $\Pi'$ such that $T(j-1, \Pi')$ is the entry that is used to compute $T(j,\Pi)$. Let $T(n, \Pi)$ be the entry that contains the maximum value over all $\Pi$. The value in this entry represents the revenue we can get from the optimal pricing $p^*$, so it is enough to reconstruct the price function $p^*$. We first obtain a sequence of profiles $\Pi^1,\ldots, \Pi^n = \Pi$ such that $\Pi^j$ is a profile for $\pset_j$ and that $\Pi^j$ is consistent with $\Pi^{j-1}$ for any $j =1,\ldots, n$. This sequence allows us to reconstruct a price function that is consistent with all the profiles in polynomial time.

\section{{\sf QPTAS} for $2$-{\sf SMP}}\label{sec:2-SMP}

In this section, we show that {\sf QPTAS} for $2$-{\sf SMP}. 

\subsection{Overview}

\danupon{This problem is harder than the highway pricing problem since it doesn't have a separator. For example, the log n approx of Balcan-Bum for Highway heavily relies on the separator. Similarly, Khaled's QPTAS also relies on the separator (once you remember the profile in the middle, you can solve two sides separately).}

We sketch the key ideas here and leave the details in next sections. First, consider the special case where each consumer has
budget $1$ or $2$ and each item must be priced either $0$ or $1$.
The exact optimal solution of this case can be found in
$n^{O(\log^2 m n)}$ time. We later
show how to extend the idea to the general cases, which turns out to
be easy for the case of highway problem but need a few more ideas
for the case of $2$-{\sf SMP}.

\paragraph{Algorithm for highway pricing problem reviewed:}
Let us first start with the highway pricing problem which can be
casted as a special case of $2$-{\sf SMP} where items are in the
form $(1,n), (2,n-1), \ldots, (n,1)$.
The main idea used in \cite{ElbassioniSZ07}, casted in our language
of ``partition tree'' (for convenience in explaining our $2$-{\sf SMP} algorithm later) is the following.\danupon{I removed [htb!] from the figure.}

\begin{figure}
\centering \scalebox{1.2}[1.2]{\includegraphics{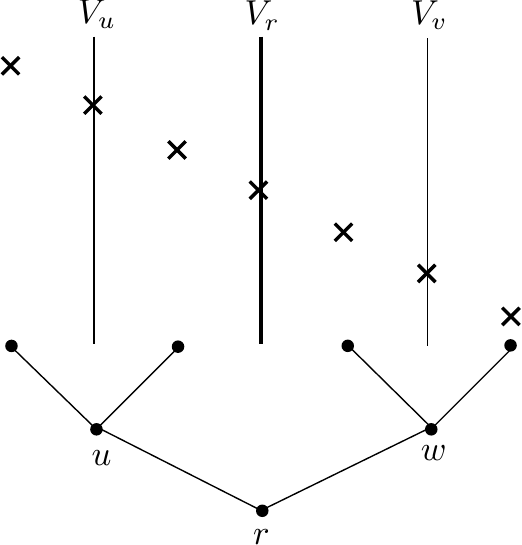}}
\caption{A partition tree}\label{fig:Partition_Tree}
\end{figure}

%
%
%
%
%
%
%

We first construct a balanced binary tree called a {\em partition tree} and denoted by
${\cal T}$. We define the vertical gridline in the middle to be a level-$0$ line, denoted by $\ell_r$, dividing the items equally to left and right sides. This line corresponds to the root node $r$ of the tree. We also assign the consumers whose consideration set contains items on both sides to the root node. Then we recursively define the subtrees on the subproblems on the two sides of line $\ell_r$ as shown in Figure~\ref{fig:Partition_Tree} until we reach the subproblem containing only one item. For any node $v \in \tset$, let $\cset_v$ be the set of consumers assigned to $v$, and $\ell_v$ be the line associated with node $v$.


Now we show a top-down recursive algorithm to solve this problem.
This algorithm can be converted to a dynamic program by working
bottom-up instead. At the root node $r$ of ${\cal T}$, we would like
to compute $f_r(\Item_{L, 1}, \Item_{L, 2}, \Item_{L, 3}, \Item_{R,
1}, \Item_{R, 2}, \Item_{R, 3})$ which is defined to be the optimal
revenue that we can collect from consumers in ${\cal C}\setminus
{\cal C}_r$ when we price the items in such a way that $\Item_{L,
1}$, $\Item_{L, 2}$ and $\Item_{L, 3}$ ($\Item_{R, 1}$, $\Item_{R,
2}$ and $\Item_{R, 3}$, respectively) are the first, second, and
third closest items on the left (respectively, right) of $\ell_r$
that have price $1$. To avoid long notation, let us denote
$\{\Item_{L, 1}, \Item_{L, 2}, \Item_{L, 3}, \Item_{R, 1}, \Item_{R,
2}, \Item_{R, 3}\}$ by $\Gamma_r$ and $f_r(\Item_{L, 1}, \Item_{L,
2}, \Item_{L, 3}, \Item_{R, 1}, \Item_{R, 2}, \Item_{R, 3})$ by
$f_r(\Gamma_r)$.
%
%
%
If we can compute $f_r(\Gamma_r)$ for all $\Gamma_r$ then the
optimal revenue can be obtained via the following formula.
\begin{align}
\text{Optimal revenue} &=
\max_{\Gamma_r} f_r(\Gamma_r) + m_1(\Gamma_r) + 2m_2(\Gamma_r)\label{eq:highway}
\end{align}
where, for any node $v$, $m_1(\Gamma_v)$ is the number of consumers
in ${\cal C}_v$ whose consideration sets contain exactly one item in
$\Gamma_v$, and $m_2(\Gamma_v)$ is the number of consumers in ${\cal
C}_v$ with budget $2$ whose consideration sets contain exactly two
items in $\Gamma_v$. The point is that we can calculate the revenue
from consumers in ${\cal C}_r$ as $m_1(\Gamma_r) + 2m_2(\Gamma_r)$
and use $f_r(\Gamma_r)$ to compute the revenue obtained from the
rest of the consumers.

It is left to compute $f_r(\Gamma_r)$. Let $u$ and $v$ be the left
and right children of $r$, respectively. In order to compute
$f_r(\Gamma_r)$, we will compute $f_u(\Gamma_r, \Gamma_u)$ which is
the maximum revenue we can collect from consumers assigned to the
descendants of $u$ (excluding $u$) where $\Gamma_r$ is the set of
six items of price $1$ that are closest to $\ell_r$ as defined
earlier. And, similarly, $\Gamma_u=\{\Item'_{L, 1}, \Item'_{L, 2},
\Item'_{L, 3}, \Item'_{R, 1}, \Item'_{R, 2}, \Item'_{R, 3}\}$ is the
set of six items of price $1$ that are closest to $\ell_u$.
Moreover, we require that $\Gamma_u$ must be {\em consistent} with
$\Gamma_r$ in the sense that there is some price assignment such
that items in $\Gamma_u$ are the items closest to $\ell_u$ of price $1$
and items in $\Gamma_r$ are the items closest to $\ell_u$ of price $1$ as
well.
%
%
(For example, if we let $\Gamma_r=\{\Item_{L, 1}, \Item_{L, 2},
\Item_{L, 3}, \Item_{R, 1}, \Item_{R, 2}, \Item_{R, 3}\}$ then an
item $\Item$ with property $\Item_{L, 3}[1]<\Item[1]<\Item_{L,
2}[1]$ cannot be in $\Gamma_u$ since this item must have price $0$.)
We use $\Gamma_u\bowtie\Gamma_r$ to denote ``$\Gamma_u$ is
consistent with $\Gamma_r$''. We define $f_v(\Gamma_r, \Gamma_v)$ in
a similar way.

Once we have $f_u(\Gamma_r, \Gamma_u)$ and $f_v(\Gamma_r, \Gamma_v)$
for all $\Gamma_u\bowtie\Gamma_r$ and $\Gamma_v\bowtie\Gamma_r$, we
can compute $f_r(\Gamma_r)$:
\begin{align}
 f_r(\Gamma_r) & = \max_{\Gamma_u\bowtie\Gamma_r} \left\{f_u(\Gamma_r, \Gamma_u) + m_1(\Gamma_u) + 2m_2(\Gamma_u)\right\} + \max_{\Gamma_v\bowtie\Gamma_r} \left\{f_v(\Gamma_r, \Gamma_v) + m_1(\Gamma_v) + 2m_2(\Gamma_v)\right\} \,.\label{eq:highway_recurse}
\end{align}
The main point here is that there is no consistency requirement between $\Gamma_u$ and $\Gamma_v$ so we have two independent subproblems. We define the function $f_z$, for all nodes $z$ in ${\cal T}$ similarly: Let $r=v_0$, $v_1$, $v_2$, ..., $v_{q-1}$ be the ancestors of $z$ and $v_q=z$. We have to compute $f_z(\Gamma_{v_0}, \Gamma_{v_1}, \ldots, \Gamma_{v_q})$ for all $\Gamma_{v_0}, \Gamma_{v_1}, \ldots, \Gamma_{v_q}$ such that $\Gamma_{v_i}\bowtie \Gamma_{v_j}$ for all $i\neq j$.

The computation of $f_z(\Gamma_{v_0}, \Gamma_{v_1}, \ldots,
\Gamma_{v_q})$ is done in the same way as Eq.\eqref{eq:highway_recurse} for every non-leaf node $z$.  At leaf node $z$, $f_z(\Gamma_{v_0}, \Gamma_{v_1}, \ldots,
\Gamma_{v_q})$ can also be easily computed: $f_z(\Gamma_{v_0},
\Gamma_{v_1}, \ldots, \Gamma_{v_q})=m_1(\Gamma_z)+2m_2(\Gamma_z)$.

Observe that $q=O(\log m + \log n)$ and there are $n^{6}$
choices for each $\Gamma_{v_i}$. Therefore, we can precompute
$f_z(\Gamma_{v_0}, \Gamma_{v_1}, \ldots, \Gamma_{v_q})$ for all
$n^{O(\log m + \log n)}$ combinations of $\Gamma_{v_0},
\Gamma_{v_1}, \ldots, \Gamma_{v_q}$. By working bottom-up from the
leaf nodes, the running time becomes $\poly(m)n^{O(\log m + \log n)}$.

\paragraph{Algorithm for $2$-{\sf SMP} (special case):} To solve the special case of $2$-{\sf SMP} defined above, we need to modify a few definitions in a right way.
%
%
Let us again consider the top-down algorithm and start at the root node $r$ of the partition tree ${\cal T}$. (Recall that we can assume that there is at most one item in each row and column so we can still define the paritition tree by drawing the vertical line through the point in the middle when sorted by the first dimension.)

One problem immediately appears: $f_r(\Gamma_r)$ cannot be used to compute the optimal revenue as we did in Eq.\eqref{eq:highway}. The reason is that we cannot compute the revenue from ${\cal C}_r$ using $m_1(\Gamma_r) + 2m_2(\Gamma_r)$ anymore.
To fix this, we have to redefine ${\cal C}_r$ in the following way:
We assign all consumers lying on the left (respectively, right) of
$\Item_r$ to the left (respective, right) child and keep only those
consumers lying exactly on the vertical line going through $\Item_r$
in ${\cal C}_r$.

Now we can compute the revenue from the newly defined ${\cal C}_r$
and a function that computes the total revenue. To do this, we
define $f_r(\Item_1, \Item_2, \Item_3)$ to be the total revenue we
can get from consumers in ${\cal C}\setminus{\cal C}_r$ by pricing
the items in such a way that, among the items on the right side of
$\Item_r$, items $\Item_1$, $\Item_2$, and $\Item_3$ are the items
with price $1$ that have the highest, second highest, and third
highest values in the second dimension, respectively. Again, let
$\Gamma_r$ denote a possible choice of $\{\Item_1, \Item_2,
\Item_3\}$ and write $f_r(\Gamma_r)$ instead of $f_r(\Item_1,
\Item_2, \Item_3)$. If we can compute $f_r(\Gamma_r)$ then we can
get the optimal revenue by Eq.\eqref{eq:highway},
%
%
where $m_1(\Gamma_r)$ and $m_2(\Gamma_r)$ is as defined earlier (with the new definition of ${\cal C}_r$).

Some more complications lie in computing $f_r(\Gamma_r)$, for any
$\Gamma_r$. As before, we will compute $f_u(\Gamma_r, \Gamma_u)$ and
$f_v(\Gamma_r, \Gamma_v)$ where $u$ and $v$ are the left and right
children of $r$, respectively. Howerver, we have to carefully define
$f_u(\Gamma_r, \Gamma_u)$ and $f_v(\Gamma_r, \Gamma_v)$, in a
different way.

We define $f_u(\Gamma_r, \Gamma_u)$, for any $\Gamma_u=\{\Item_1,
\Item_2, \Item_3\}$, to be the maximum revenue from the consumers
assigned to the descendants of $u$ when we price the items in such a
way that, among the items lying on the right side of $\Item_u$ and
left side of $\Item_r$, items $\Item_1$, $\Item_2$, and $\Item_3$
are the items with price $1$ that have the highest, second highest,
and third highest values in the second dimension, respectively. Note
that we do not need to check any consistency between $\Gamma_r$ and
$\Gamma_u$: For any choice of $\Gamma_r$ and $\Gamma_u$, there is
always a price assignment such that items in $\Gamma_r$ and
$\Gamma_u$ are the items of price $1$ that have the highest values
in the second dimension in their respective regions. In this case,
we say that $\Gamma_r\bowtie \Gamma_u$ is always true for any
$\Gamma_r$ and $\Gamma_u$.

On the other hand, we define $f_v(\Gamma_r, \Gamma_v)$, for any
$\Gamma_v=\{\Item_1, \Item_2, \Item_3\}$, to be the maximum revenue
from the consumers assigned to the descendants of $v$ when we price
the items in such a way that, among the items lying on the right
side of $\Item_v$, items $\Item_1$, $\Item_2$, and $\Item_3$ are the
items with price $1$ that have the highest, second highest, and
third highest values in the second dimension, respectively. In this
case, we have to make sure that $\Gamma_v$ is consistent with
$\Gamma_r$, i.e., there is some price assignment such that items in
$\Gamma_r$ and $\Gamma_u$ are the items of price $1$ that have the
highest values in the second dimension in their respective regions.

Now we have defined $f_u(\Gamma_r, \Gamma_u)$ and $f_v(\Gamma_r,
\Gamma_v)$, we compute $f_r(\Gamma_r)$ using
Eq.\eqref{eq:highway_recurse}.
%
%
As in the case of the highway pricing problem, we can extend the
definition to other nodes. In particular, at a leaf node $z$ we have
to compute $f_z(\Gamma_{v_0}, \Gamma_{v_1}, \ldots, \Gamma_{v_q})$
where $q=O(\log m + \log n)$. Hence, this case can be solved in
$\poly(|{\cal C}|)\cdot\size{\cal I}^{\polylog{|{\cal I}|}}$ time.

\danupon{To do: try not to use $m$ and $n$.}

\paragraph{Algorithm for general $2$-{\sf SMP}:}
We now remove the restrictions that each item must be priced $0$ or
$1$ and each budget must be $1$ or $2$. The removal of the
restriction on item price does not affect the case of highway
pricing problem since this can be easily assumed (see, e.g.,
\cite{GrandoniR10}).\danupon{Actually, can we assume this? In
\cite{GrandoniR10}, they also assume this for tollbooth problem on
trees (Section 4.1).}
Moreover, we can still assume that the maximum budget is $O(m n)$. Now we can deal with the general highway problem by
redefining $f_r(\Gamma_r)$: Let $\Gamma_r=\{\Item_{L,0}, \Item_{L,
1}, \ldots, \Item_{L, q}, \Item_{R,0}, \Item_{R, 1},
\ldots,\Item_{R, q}\}$ where $q=O(\log m n)$.
For any $i\leq q$, we want to price in such a way that $\Item_{L,
i}$ is the item closest to $\Item_r$ on the left such that the sum
of the price of all items between $\Item_r$ and $\Item_{L, i}$ is at
least $(1+\epsilon)^i$. Computing $f_r(\Gamma_r)$ can be done in the
same manner as before and consistency checking is easy to deal with.
Function $f_{v_q}(\Gamma_{v_0}, \Gamma_{v_1}, ..., \Gamma_{v_q})$,
for any node $v_q$ at level $q$ in $\cal T$, can be defined in a
similar manner.

For $2$-{\sf SMP}, we may not in general assume the item prices to
be $0/1$. Instead, we show that it can be assumed that each item
must have price $0$, or $(1+\epsilon)^j$, for any $j=0, 1, \ldots,
O(\log m)$. A natural extension of the above idea is to
define the notion of ``volume of regions'': For each item $\Item$,
let $H_{\Item}$ and $V_{\Item}$ denote the horizontal and vertical
line cutting through item $\Item$, respectively. Any rectangle
resulting from drawing some horizontal and vertical lines through
some items are called {\em regions} and the regions that do not
contain other regions are called {\em minimal regions}. For any
price assignment, we define the {\em volume} of a region to be the
sum of the price of all items within the region.

\begin{figure}
\centering \scalebox{0.7}[0.7]{\includegraphics{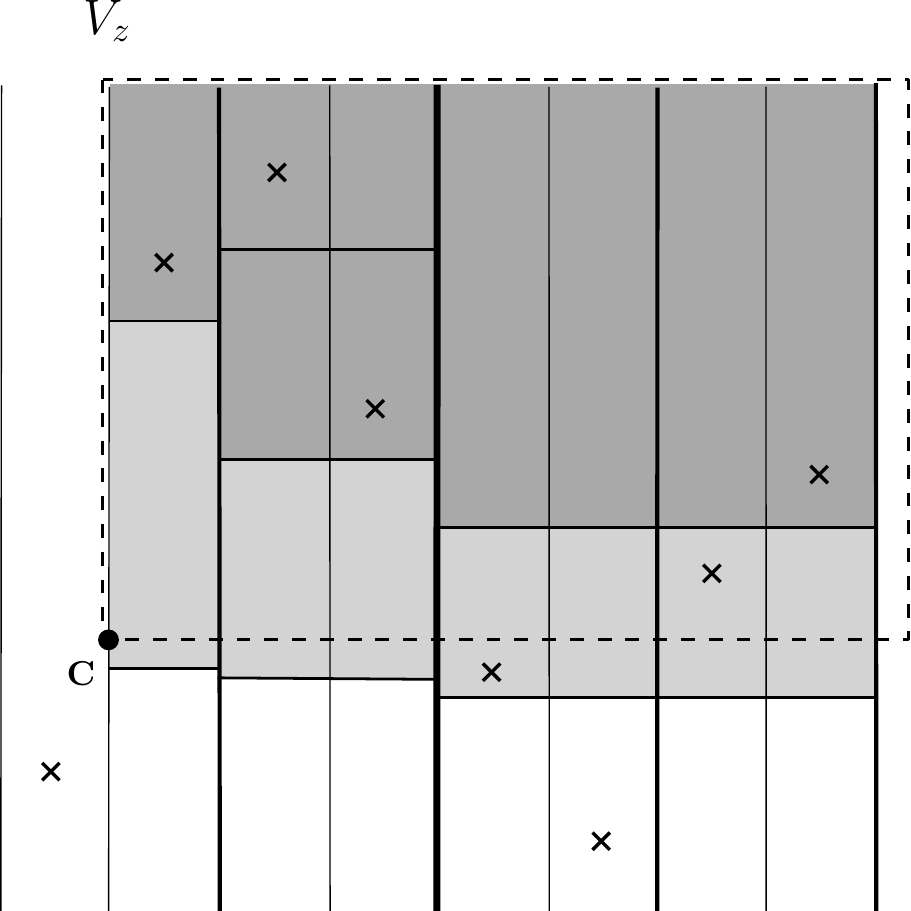}}
\caption{Approximating the revenue from consumer $\Consumer$ assigned to node  $z$ in $\mathcal{T}$.}\label{fig:area-idea}
\end{figure}

%
%

Now, similar to the highway problem, we define $\Gamma_r=\{\Item_0,
\Item_1, ..., \Item_k\}$ (note that $k=O(\log m)$) as the
``region guess'': We define $f_r(\Gamma_r)$ to be the maximum
revenue from ${\cal C}\setminus {\cal C}_r$ when we price in such a
way that, for any $i$, item $\Item_i$ is the highest item (in the
second dimension) such that the volume of the region on the right of
the vertical line $V_{\Item_r}$ and above the horizontal line
$H_{\Item_i}$ (including $\Item_i$) is at least $(1+\epsilon)^i$.
Using these volume guesses, we can approximate the upper and lower
bounds of the revenue from each consumer $\Consumer$ at node $z$ by
looking at $\Gamma_v$ for all ancestors $v$ of $z$. This is because
each consumer's consideration set will contain some set of regions
$B_1, B_2, ...$ with volume guesses $(1+\epsilon)^{i_1},
(1+\epsilon)^{i_2}, ...$, respectively (such as the blue regions in
Figure~\ref{fig:area-idea}). Also, this consideration set will also
be contained in some set of regions $R_1, R_2, ...$ with volume
guesses $(1+\epsilon)^{i_1+1}, (1+\epsilon)^{i_2+1}, ...$ (such as
the blue and red regions together in Figure~\ref{fig:area-idea}).

However, in contrast to the case of highway problem, the consistency
between the guesses (e.g., between $\Gamma_r$ and its children
$\Gamma_u$ and $\Gamma_v$) is harder to guarantee. In order to
guarantee the consistency, we add another parameter, denoted by
$\Delta_r=\{\delta_0, \delta_1, \ldots, \delta_{O(k^2)}\}\subseteq
R_{\geq 0}^{O(k^2)}$ (recall that $|\Gamma_r|=k+1$). $\Delta_r$ is
used as a ``volume guess''. That is, we define $f_r(\Gamma_r,
\Delta_r)$ to be the maximum revenue from ${\cal C}\setminus {\cal
C}_r$ when we price in such a way that the restriction on $\Gamma_r$
is as before and, additionally, the volumn of the $i$-th minimal
region is exactly $\delta_i$ (where we make any order of the minimal
regions). We can now guarantee the consistency by making sure that
the sum of the volume guesses in smaller regions defined by
$\Gamma_u$ and $\Delta_u$ (as well as $\Gamma_v$ and $\Delta_v$) is
exactly the volume guesses in the bigger regions defined by
$\Gamma_r$ and $\Delta_r$. \danupon{Need to make this part much more
precise later. Also, we have to make it very clear why we need the
volume guesses, not region guesses alone.}

For any node $z$, we also define a function $f_z(\Gamma_{v_0},
\Gamma_{v_1}, \ldots, \Gamma_{v_q}, \Delta_z)$ where $v_0, v_1,
\ldots, v_{i-1}$ are ancestors of $z$ and $v_q=z$. In this case, we
consider the minimal regions obtained by drawing vertical lines
$V_{\Item_{v_0}}, V_{\Item_{v_1}}, \ldots, V_{\Item_{v_q}}$ and
horizontal lines $H_{\Item}$ for $\Item\in \Gamma_{v_i}$, for all
$i$. We use $\Delta_z$ to store the numbers that are the ``volume
guesses'' of all these regions. We also check the consistency in
terms of volume, i.e., $\Pi=\{\Gamma_{v_0}, \Gamma_{v_1}, \ldots,
\Gamma_{v_q}, \Delta_{v_q}\}$ is consistent with
$\Pi'=\{\Gamma_{v_0}, \Gamma_{v_1}, \ldots, \Gamma_{v_{q-1}},
\Delta_{v_{q-1}}\}$ if the volume guesses of the smaller regions
defined by lines in $\Pi$ add up to the volume guesses of the bigger
regions defined by lines in $\Pi'$.

\subsection{Preprocessing}
Fix some $\epsilon >0$. Given an instance $(\iset, \cset)$, our goal
is to compute a price that collects a revenue of at least
$(1-O(\epsilon)) \opt$. Recall that we can assume that the consumers
are on the intersection of grid lines, and the items are in the grid
cells (cf. Lemma~\ref{lem:perturb}). First we process the input so
that the budgets and prices are polynomially bounded. Moreover, the
optimal solution only assigns prices of the form $(1+\epsilon)^j$
for some $j \leq O(\log m)$. The proof of this fact only uses
standard arguments (along the same line as in \cite{BalcanB07}).


\begin{lemma}\label{lem:preprocess}
Let $M= O(mn /\epsilon)$. The input instance ${\cal P}$ can be
reduced to ${\cal P'}$ with the following properties.

\squishlist
  \item For each consumer ${\bf C}$, the budget of ${\bf C}$ in ${\cal P}'$ is between $1$ and $M$.
  \item Any price $p'$ that $\alpha$-approximates the optimal pricing of $\pset'$ can be transformed in polynomial time into another price $p$ that gives $(1+3\epsilon) \alpha$-approximation for $\pset$.
  \item There is a $(1+\epsilon)$-approximate solution $\tilde{p}$ satisfying the following: For all ${\bf I} \in \iset$, $1 \leq \tilde{p}({\bf I}) \leq M$, and $\tilde{p}({\bf I})$ is in the form $(1+\epsilon)^j$ for some $j \leq O(\log m)$.
\squishend
\end{lemma}
\begin{proof}
Let $B_{\max}$ be the maximum budget among all consumers. We first
remove all consumers whose budgets are less than $\epsilon
B_{\max}/mn$. Notice that we only lose the revenue of at most
$\epsilon B_{\max} \leq \epsilon \opt$ by this removal. We denote
the new set of consumers by $\cset'$. Now look at the optimal price
$p^*$ for the resulting instance. If for some ${\bf I} \in \iset$,
the price $p^*({\bf I})$ is less than $\epsilon B_{\max}/mn$, we
change its price to $p'({\bf I}) = 0$ and remove item ${\bf I}$
completely from the instance. Again, since each such item can only
be sold to at most $m$ consumers, discarding it only decreases the
revenue by $\epsilon B_{\max}/n$. There are at most $n$ such items,
so we lose a revenue of at most $\epsilon \opt$ in total. Let
$\iset'$ denote the resulting set of items.

Next we scale each consumer budget by $M'=mn/\epsilon B_{\max}$ to
get a new budget, i.e. $B'_C = M'B_C$. Now we have a complete
description of the instance $\pset'$ in which consumer budgets are
between $1$ and $M$. Let $\opt'$ be the optimal value of the new
instance. First we try to lower bound the value of $\opt'$. Consider
the same price $p^*: \iset' \rightarrow \R$ scaled up by a factor of
$M'$. The revenue from this price is at least $(1 - 2\epsilon)
M'\opt$, so we have that $\opt' \geq (1-2\epsilon)M'\opt$.

We are now ready to prove the second part. Assume that we have a
price $p'$ that gives $\alpha$-approximation for $\pset'$, so the
revenue collected by $p'$ is at least $\opt'/\alpha$. We construct
the price $p$ by scaling down the price of $p'$ by $M'$. Notice that
for each consumer ${\bf C}$ who can afford his consideration set in
$\pset'$ with price $p'(S_C)$, he can also afford his set in $\pset$
with price $p'(S_C) = p(S_C)/M'$. Therefore, the revenue collected
by $p$ is at least $\opt'/\alpha M' \geq (1-2\epsilon)\opt/ \alpha$.
This argument also implies that $\opt \geq \opt'/M'$.

Finally we show that there is a good solution $\tilde p$ that only
assigns prices in the form $(1+\epsilon)^j$, as follows. We round
down the price of $p^*$ to the nearest scale of $(1+\epsilon)^j$ for
some $j$. For each consumer $\Consumer$ who purchases item ${\bf I}$
w.r.t. price $p^*$, by scaling down every item price, she can still
afford her consideration set $S_{\Consumer}$, whose new price is at
least $p^*(S_{Consumer})/(1+\epsilon) \geq (1-\epsilon)
p^*(S_{\Consumer})$.
\end{proof}


From now on, we assume that our input instance and its optimal price
are in such format. Our goal is to devise a {\sf QPTAS} for this
instance. We note here that in some special cases of single-minded
pricing problems, especially the Highway problem, an even stronger
statement can be assumed; namely, that the optimal price is
integral~\cite{GuruswamiHKKKM05}. It seems that such a nice property
may not hold in our case, and we anyway do not need it.

\subsection{Partition tree}

We first construct a (almost balanced) binary tree $\tset$ where
each node in $\tset$ is associated with a rectangular region in the
plane (from now on, whenever we talk about region, we always mean a
rectangular one). We call this tree the {\em partition tree}. It can
be constructed recursively as follows. In the beginning, we have
$\tset=\set{r}$ where $r$ is the root of the tree whose region $A_r$
is the whole grid. We repeat the following process: For each leaf $v
\in \tset$, if the region $A_v$ of $v$ contains at least two items,
we choose a vertical grid line $\ell_v$ dividing the items in a
balanced manner to the left and right side. We then add the left
child $v'$ of $v$ with the region $A_{v'}$ being the region of $A_v$
on the left side of $\ell_v$. We also add the right child $v''$ of
$v$ associated with the region $A_{v''}$ on the right side of
$\ell_v$. We repeat the process until every leaf is associated with
a region containing only one item; see Figure~\ref{figure: tree}.

\danupon{It's better to define ``volume'' later because we need to
define ``region'' with horizontal lines first.} For each node $v \in
\tset$, we define the item set $\iset_v$ to be the set of all items
in the region $A_v$. Fix a price $p: \iset \rightarrow \R$. For any
region $A$, we define the ``volume'' $\vol_p(A)$ to be the total
price among all items in the region, i.e. $\vol_p(A) = \sum_{{\bf I}
\in A} p({\bf I})$. The following simple claim
is crucial in designing our algorithm.

\begin{claim}\label{claim:bound_sum}
Let $p^*$ be an optimal price. Then for any region $A$, there are
only $n^{O(\log m)}$ possible values of $\vol_{p^*}(A)$.
\end{claim}

\begin{proof}
Let $x_j$ denote the number of items ${\bf I}$ in $A$ with price
$p^*({\mathbf I}) = (1+\epsilon)^j$. Notice that we can write the
volume of $A$ as $\sum_{j=1}^{q} x_j (1+\epsilon)^j$ where $x_j$
only takes non-negative integer values at most $n$. So we have at
most $n^{O(\log m)}$ possibilities for the volume.
\end{proof}

\subsection{Horizontal partition and local profile}

From the construction, each node $v$ of the partition tree, is
associated with a vertical line $\ell_v$ which divides the plane
into two region. We further partition the right region using
vertical line, as follows.

Consider a non-leaf node $v \in \tset$ with left child $v'$ and
right child $v''$. A {\em horizontal partition} for node $v$,
denoted by $H_v$, is a collection of (not-necessarily distinct)
horizontal lines $\ell^v_1,\ldots, \ell^v_q$, partitioning the
region of $A_{v''}$ into many pieces; note that the left endpoints
of these lines are on $\ell_v$.\danupon{Need a picture here.} The
line $\ell^v_j$ is supposed to mark the highest $y$-coordinate such
that the volume inside $A_{v''}$ above $\ell^v_j$ is at least
$(1+\epsilon)^j$. Notice that each node $v$ has at most $n^{O(\log
m)}$ feasible partitions since there are at most $n$ possibilities
for the choice of each $\ell^v_j$.

Now if we fix a horizontal partition of every non-leaf node in the
partition tree, we can define {\em minimal} regions for each
non-leaf node $v$ as follows. For each node $v$, we consider all
vertical and horizontal lines associated with $v$ and all its
ancestors (i.e., all lines in $\ell_u$ and $H_u$ where $u=v$ or $u$
is an ancestor of $v$). Let $\lset_v$ denote the set of these lines.
$\lset_v$ naturally defines minimal regions: We say that a region
$A$ is minimal with respect to $\lset_v$ if and only if $A$ is a
rectangle whose four boundaries are the lines in $\lset_v$, and
there is no line in $\lset_v$ that intersects with the interior of
$A$.


Now, we define a {\em local profile} of a node $v$. It consists of
(i) horizontal partitions for $v$ and for all its ancestors, and
(ii) numbers on every minimal region resulting from vertical and
horizonal lines. The numbers are supposedly the ``volume guesses''
of every minimal region of $v$.

Now we try to guess the ``right'' local profile of every node in the
partition tree. We show that if this guess is right, then we get a
good approximation of the optimal solution. Moreover, we can use
dynamic programming to make the right guess.

\subsection{Dynamic Programming Solution}
%

%


A {\em global profile} (or just {\em profile} in short) of a node
$v$ consists of the local profile of $v$ and all its ancestors in
such a way that the volumes of minimal regions of $v$ is consistent
with its ancestors. More formally, fix a node $v$. A profile $\Pi_v$
for $v$ consists of, for any ancestor $v'$ of $v$, $\Pi_{v, v'}$
which is the local profile that node $v$ wants its ancestor $v'$ to
have (we also think of $v$ has an ancestor of itself for
convenience).
As a reminder, for each ancestor $v'$ of $v$, local profile $\Pi_{v,
v'}$ some horizontal partition $H_{v'}$ and the ``volume guess''of
each minimal region of $v$.
%
%
%
%
In addition, we restrict that these local profiles $\Pi_{v, v'}$
have to be consistent in themselves in the following sense. For each
vertex $v'$, for any minimal region $A'$ of $\Pi_{v,v'}$ that is
further partitioned into minimal regions $A'_1, A'_2,\ldots,
A'_{\gamma}$ of $\Pi_{v,v''}$ for some descendant $v''$ of $v'$, the
number $z_{A'}$ at $\Pi_{v,v'}$ is equal to the sum of the numbers
$z'_{A_j}$ of $\Pi_{v,v''}$.\danupon{This is still not clear and a
bit informal.}

We argue that the number of global profiles for each node is not too
large, i.e. only $n^{\operatorname{poly} \log m}$. There are
$n^{O(\log m)}$ horizontal partitions for each ancestor $v'$ of $v$,
making a total of $n^{O(\log m \log n)}$ possibilities of the lines
$\ell^{v'}_j$. Now fix a choice of such horizontal partitions. If we
draw all lines $\ell^{v'}_j$ involved in the global profiles, we
will see a number of regions formed by intersections between these
lines and the vertical lines $\ell_{v''}$. Since there are $O(\log m
\log n)$ such horizontal lines and $O(\log n)$ vertical lines
involved, we have at most $O(\log m \log^2 n)$ minimal rectangular
regions. Each region has at most $n^{O(\log m)}$ possible volumes,
so there are at most $n^{O(\log^2 m \log^2 n)}$ global profiles for
each node $v$\parinya{We implicitly used the fact that $m \leq O(n)$. Have to say
it somewhere}.

Now we define a {\em valid tree profile} $\Pi$ for $\tset$ as the
set of global profiles $\set{\Pi_v}_{v \in \tset}$ such that $\Pi_v$
is a global profile for node $v$. Moreover, for every parent-child
pair $v, v'$ where $v$ is a parent of $v'$ in $\tset$, the profile
$\Pi_{v'}$ agrees with $\Pi_v$. That is, all profiles about
ancestors of $v$ in $\Pi_v$ and $\Pi_{v'}$ are exactly the same.

Given a valid tree profile $\Pi$, we have the notion of cost of the
profile $\Pi$ (denoted by $\mbox{\sf Cost}(\Pi)$) which is supposed
to approximate the total revenue we can collect by a price function
consistent with $\Pi$. The cost of a profile can be computed as
follows. For each node $v \in \tset$, let $\cset_v$ be the set of
all consumers on line $\ell_v$. For each consumer ${\bf C} \in
\cset_v$, the rectangular region enclosed by horizontal line ${\bf
C}[2]$ and vertical line $\ell_v$ 
is the actual amount the consumer needs to pay. This is the amount
we do not know, but we can approximate: We let $v_0, v_1,\ldots,
v_{\alpha}$ be a sequence of ancestors of $v$ such that $v$ is on
the left subtree of $v_i$ (in the order from $v$ to the root), where
$v_0= v$. And we let for each $i$, $j_i$ be the maximum number such
that $\ell^{v_i}_{j_i}$ does not lie below ${\bf C}[2]$. The cost of
consumer ${\bf C}$ is just the sum $\sum_{i=0}^{\alpha}
(1+\epsilon)^{j_i}$ if $B_C \leq \sum_{i=0}^{\alpha}
(1+\epsilon)^{j_i}$ and zero otherwise. The cost at node $v$ is just
the total cost of all consumers in $\cset_v$, and the cost of the
profile is the sum of the cost over all nodes $v \in \tset$.

\begin{figure}
\centering
\subfigure[Tree Decomposition up to depth $2$ where the shaded
region denotes $A_v$]{
\includegraphics[height=5cm]{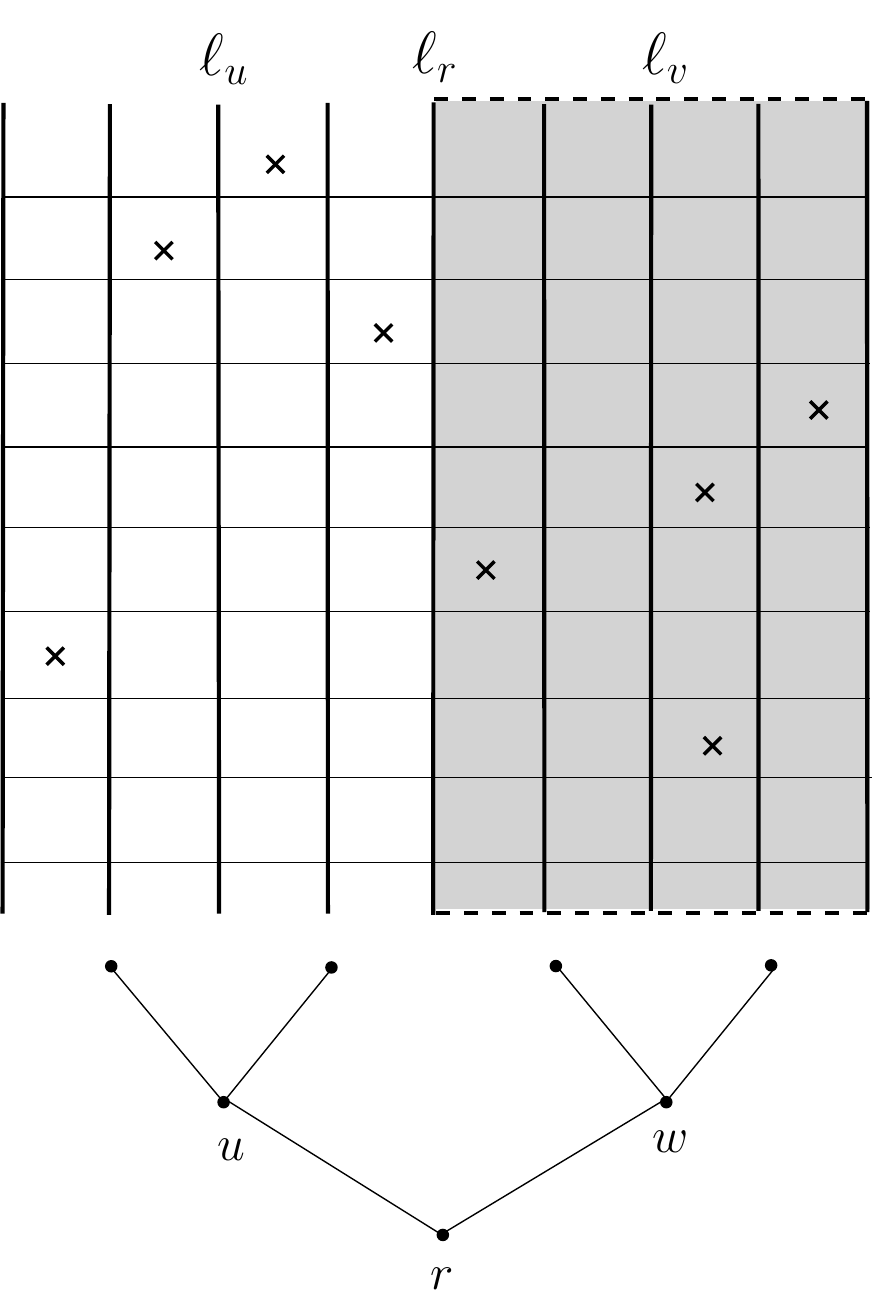}\label{figure:
tree} }
\hspace{.05\textwidth}
\subfigure[The actual volume]{
\includegraphics[height=5cm]{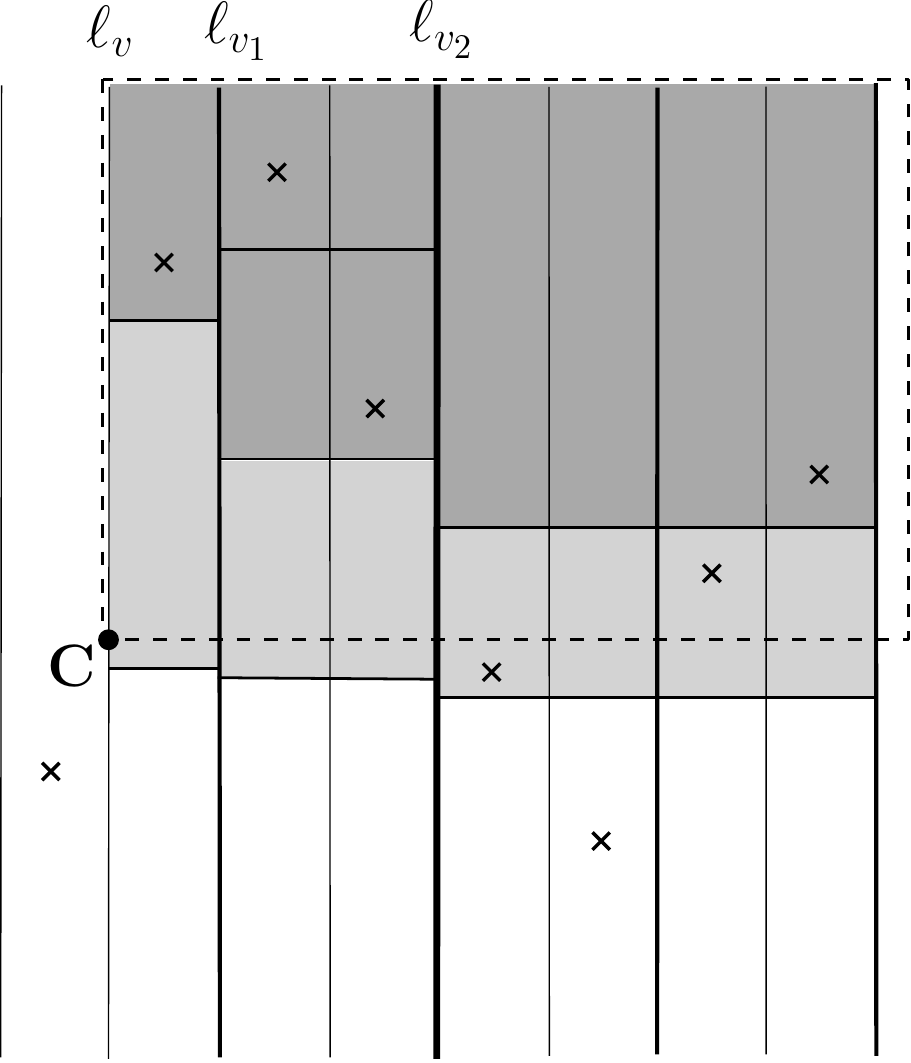}
\label{fig:area1} }
\hspace{.05\textwidth}
\subfigure[The approximate volume]{
\includegraphics[height=5cm]{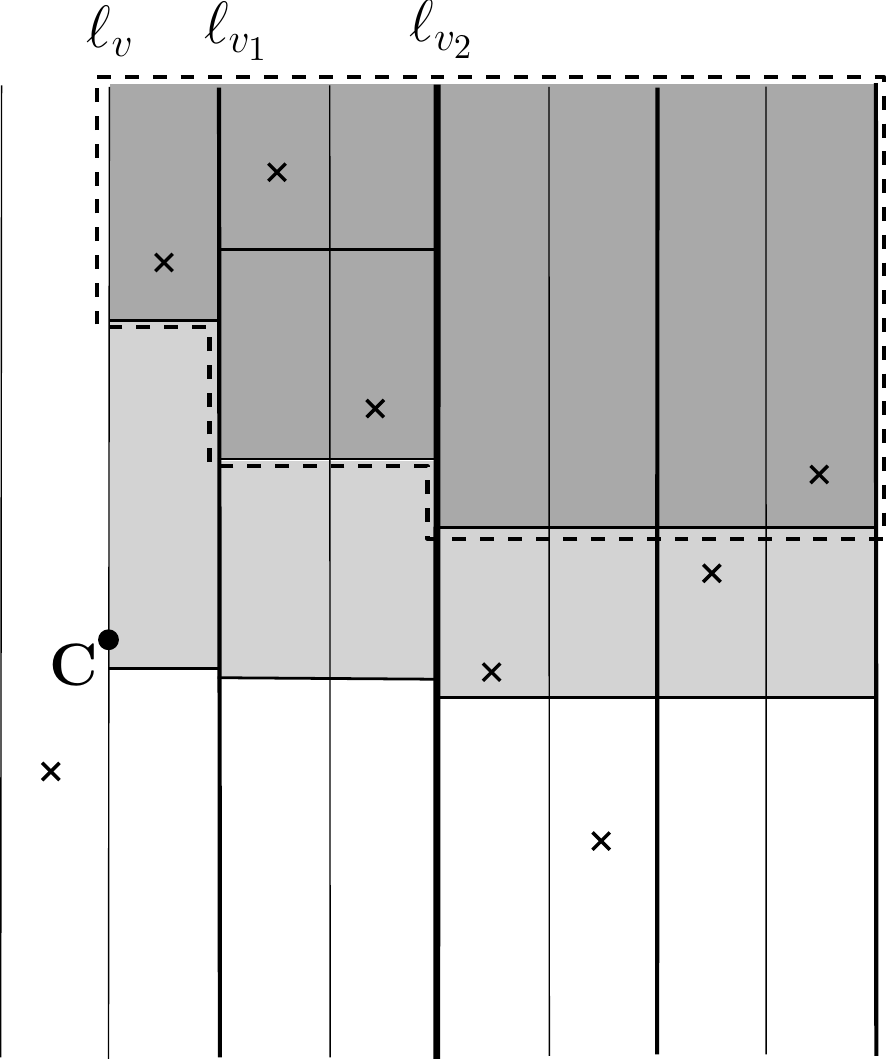}
\label{fig:area2}}
\caption{Computing the cost of consumer ${\bf C}$}\label{fig:area}
\end{figure}

\begin{lemma}\label{lem:exist_good_tree}
There is a valid tree profile $\Pi^*$ such that the cost is at least
$(1-\epsilon)\opt$.
\end{lemma}

\begin{proof}
We start from the optimal price $p^*$ and construct the valid
profile as follows. For each node $v$, we define a feasible
partition of $v$ by choosing the line $\ell^v_j$ to be at the
highest $y$-coordinate such that the total volume enclosed is at
least $(1+\epsilon)^j$. Then we create a profile $\Pi^*_v$ for each
node $v$ according to the actual volume of each minimal region.
Notice that this gives a valid tree profile.
\end{proof}

Our goal now is to compute the valid profile $\Pi$ of maximum cost
by dynamic programming, and the profile will automatically suggest a
near-optimal pricing.

\paragraph{Computing the Solution:} Let $v \in \tset$.
We say that a price $p: \iset_v \rightarrow \R$ is consistent with
global profile $\Pi_v$ if and only if for every minimal region $A$
of $\Pi_v$ that is completely contained in $A_v$, we have $\vol_p(A)
= z_A$. The minimum cost profile can be computed in a bottom-up
fashion, as follows. For a leaf node $v$, a global profile for $v$
automatically determines the price of the only item in $A_v$;
discard a profile which does not have consistent price.

The following lemma shows that a price $p$ consistent with a valid
tree profile $\Pi$ can be computed from $\Pi$.

\begin{lemma}\label{lem:consistent}
For each node $v$ with left child $v'$ and right child $v''$, let
$p':\iset_{v'}\rightarrow \R$ and $p'':\iset_{v''}\rightarrow \R$ be
the prices that are consistent with the profile $\Pi_{v'}$ and
$\Pi_{v''}$ respectively. Then the price $p: \iset_v \rightarrow \R$
defined to agree with $p'$ on $\iset_{v'}$ and with $p''$ on
$\iset_{v''}$, is consistent with $\Pi_v$.
\end{lemma}

\begin{proof}
Consider a minimal region $A \sse A_v$ and a volume guess $z_A$ in
$\Pi_v$. If $A \sse A_{v'}$ where $A$ is the union of minimal
regions $A'_1,\ldots A'_{\gamma}$ of $\Pi_{v'}$ (similar argument
can be made in case $A \sse A_{v''}$), then by assumption that
$\Pi_v$ is consistent with $\Pi_{v'}$, we know that the total value
$z_A = \sum_{j=1}^{\gamma} z'_{A'_j}$. Since $p'$ is consistent with
the profile $\Pi_{v'}$, we have that $\vol_{p}(A) = \vol_{p'}(A) =
\sum_{j} \vol_{p'}(A'_j) = \sum_j z'_{A'_j} = z_A$ as desired.
\end{proof}

We have shown that a valid tree profile $\Pi$ always has a price $p$
consistent with it. The following lemma 
basically says that this price $p$ gives a revenue close to the cost
of the profile, which will in turn imply that the maximum cost
profile gives the revenue of at least $(1-O(\epsilon)) \opt$.

\begin{lemma}\label{lem:tree-revenue}
For any valid tree profile $\Pi$, let $p$ be a price consistent with
$\Pi$ and let $p' = p/(1+\epsilon)$. Then $p'$ collects revenue at
least $(1-\epsilon)$ fraction of the profile cost.
\end{lemma}

\danupon{Still need to prove this lemma.}




\section{Omitted hardness results}\label{sec:omitted_hardness}

\subsection{Hardness of 3-{\sf UUDP-MIN} and 4-{\sf UUDP-MIN}}

In this section we show that $3$-{\sf UUDP-MIN} is \NP-hard, and
$4$-{\sf UUDP-MIN} is \APX-hard by a reduction from {\sf Vertex Cover}. Our
reduction relies on the concepts of adjacency poset and its
embedding into Euclidean space. We describe basic terminologies
here. Given a graph $G=(V,E)$, an adjacency poset $(V \cup E,
\preceq_G)$ of graph $G$ can be constructed as follows: First we
define a poset with its maximal elements corresponding to vertices in $V$ and its minimal
elements corresponding to edges $E$. For each vertex $v$ and each edge $e$, we
have the relation $e \preceq_G v$ if and only if vertex $v$ is an endpoint of
$e$. We say that a map $\phi: V \cup E \rightarrow \R^d$ is an {\em
embedding} of adjacency poset $(V \cup E, \preceq_G)$ into $\R^d$
if and only if it preserves the relations $\preceq_G$, i.e., for any two
elements $a,b \in V \cup E$, we have that $a \preceq_G b$ iff $\phi(a)[i] \leq \phi(b)[i]$ for all coordinates $i \in [d]$.

Now we describe our reductions. Since two reductions are essentially the same, we give a general procedure which will imply both results. Given an instance $G=(V,E)$ of {\sf Vertex Cover}, we first construct an adjacency poset $(V \cup E, \preceq_G)$ for $G$,
and then we compute the embedding $\phi$ of this poset into Euclidean space $\R^d$. We will use the graph $G$, as well as the embedding $\phi$, to define the instance of $d$-{\sf UUDP-MIN} as follows:

\begin{itemize}
\item {\bf Consumers:} We have two types of consumers, i.e. the rich consumers and the poor ones. For each vertex $v \in V$, we create a {\em rich} consumer $C_v$ with budget $2$ at coordinates $\phi(v)$. For each edge $e \in E$, we create a {\em poor} consumer $C_e$ with budget $1$ at coordinates $\phi(e)$.

\item {\bf Items:} For each vertex $v \in V$, we create item ${\bf I_v}$ at coordinates $\phi(v)$.
\end{itemize}
Note that for each $e=(u,v)$, each poor consumer $C_e$ has $S_{C_{e}} = \set{\Item_v, \Item_u}$, while each rich consumer $C_v$ has $S_{C_v} = \set{\Item_v}$. We denote the resulting instance by $(\cset, \iset)$.

The following lemma gives a characterization of the optimal solution for $(\cset, \iset)$. It says that we may assume without loss of generality that every poor consumer gets some item.

\begin{lemma}
For any price $p$ that is a solution for $(\cset, \iset)$ constructed above, we can transform $p$ to $p'$ such that every poor consumer buys some item with respect to $p'$, and the revenue of $p'$ is at least as much as the revenue of $p$.
\end{lemma}

\begin{proof}
Consider edge $e=(u,v)$. Suppose poor consumer $C_e$ does not get
any item, so it implies that both items ${\bf I}_u$ and ${\bf I}_v$
have price $p({\bf I}_u)= p({\bf I}_v) =2$ (recall that, since budgets are $1$ or $2$, the optimal prices would never set prices that are not in $\set{1,2}$). We define the price function $p'$ by setting
$p'({\bf I}_u) = 1$ while $p'({\bf I}_w) = p({\bf I}_w)$ for all
other vertices $w \in V \setminus \set{u}$. The only rich consumer that gets affected is $C_u$, whose payment may decrease by one. However, we earn the revenue of one back from poor consumer $C_e$. For $e' \in E: e'\neq e$, poor consumer $C_{e'}$ is never affected because his budget is one. Overall, changing the price from $p$ to $p'$ never decreases revenue.
\end{proof}

Let $p^*$ be the optimal price for $(\cset, \iset)$ and ${\sf VC}(G)$ denote the size of
minimum vertex cover of $G$. We show the following connection
between the size of minimum vertex cover and the optimal revenue
collected by $p^*$.

\begin{theorem}
The optimal revenue collected by $p^*$ is exactly $2n-{\sf VC}(G)
+m$.
\end{theorem}
\begin{proof}
From the previous lemma, we can assume that the pricing $p^*$ sells
items to every poor consumer. In other words, if $V' = \set{v:
p^*({\bf I}_v)=1}$, it must be the case that $V'$ is a vertex cover:
otherwise, let $e=(u,w)$ be an edge which is not covered by any
vertex in $V'$, so $C_e$ is only interested in items with price $2$,
which he cannot afford. This contradicts the assumption that $p^*$
sells items to every poor consumer.

The revenue collected from poor consumers is exactly $m$. Each rich
consumer $C_v$ in the vertex cover gets the item with price $1$ while
others get the items with price $2$, so the total revenue is $m + {\sf
VC}(G) + 2(n - {\sf VC}(G))$.
\end{proof}

This theorem immediately implies the gap between {\sc Yes-Instance} and {\sc No-Instance} for $d$-{\sf UUDP-MIN}. The only detail we left out is the computation of the embedding $\phi$, and this is where the hardness proofs
of $3$-{\sf UUDP-MIN} and $4$-{\sf UUDP-MIN} depart (other steps are
exactly the same). For $3$-dimensional case, we start from planar
graphs whose adjacency poset can be embedded into $\R^3$. Since
planar vertex cover has a polynomial-time approximation scheme, we
only get \NP-hardness here. For $4$-dimensional case, we start from
vertex cover in cubic graphs, which is known to be \APX-hard, but
unfortunately we can only embed its adjacency poset into $\R^4$,
thus obtaining the hardness of $4$-{\sf UUDP-MIN}.

\paragraph{\NP-Hardness of $3$-{\sf UUDP-MIN}}

To show the \NP-hardness, we start from {\sf Vertex Cover} in planar graphs, which is known to be \NP-complete~\cite{DBLP:journals/tcs/GareyJS76}. We will use the following theorem,
due to Schnyder~\cite{Schnyder89}.

\begin{theorem}
Let $(V \cup E, \preceq_G)$ be an incident poset of planar graph
$G$. Then there exists an embedding $\phi$ from the poset into
$\R^3$.
\end{theorem}

Schnyder shows later that the crucial step in the theorem
can be computed in polynomial time \cite{DBLP:conf/soda/Schnyder90}, which
immediately implies the following theorem.

\begin{theorem}
$3$-{\sf UUDP-MIN} is \NP-hard even when the consumer budgets are
either $1$ or $2$.
\end{theorem}

\paragraph{\APX-Hardness of $4$-{\sf UUDP-MIN}}

We will be using the fact that {\sf Vertex Cover} in cubic graphs
is \APX-hard \cite{DBLP:conf/ciac/AlimontiK97}, stated in the
language convenient for our use below.

\begin{theorem}
For some $0 < \alpha < \beta < 1$, it is \NP-hard to distinguish
between (i) the graph that has a vertex cover of size at most
$\alpha n$, and (ii) the graph whose minimum vertex cover is at
least $\beta n$.
\end{theorem}

Now we assume that our input graph $G$ is a cubic graph and use the
following theorem to embed the adjacency poset of $G$ into $\R^4$.

\begin{theorem}[Schnyder]
An adjacency poset of any $4$-colorable graph can be embedded into
$\R^4$. Moreover, the embedding is computable in polynomial time.
\end{theorem}

It only requires a straightforward computation to prove the
following theorem.
\begin{theorem}
$4$-{\sf UUDP-MIN} is \APX-hard even when the consumer budgets are
either $1$ or $2$.
\end{theorem}
\begin{proof}
In the \yi, we can collect the revenue of $(2-\alpha) n +m$.
However, in the \ni, the revenue is at most $(2-\beta)n +m$. Since
the graph is cubic, we may assume that $m =\gamma n$ for some $1\leq
\gamma <2$. Hence we have a gap of
$(2-\alpha+\gamma)/(2-\beta+\gamma)$.
\end{proof}

\subsection{\NP-hardness of $2$-{\sf SMP}}

Highway problem can be defined as follows: We are given a line $P = (v_0,\ldots, v_n)$ consisting of $n$ edges and $n+1$ vertices and a set of consumers $\cset$ where each consumer $\Consumer$ corresponds to a subpath $P_{\Consumer}$ of $P$ and is equipped with budget $B_\Consumer$. Our goal is to set price to edges so as to maximize the revenue, where each consumer $\Consumer$ buys a path $P_{\Consumer}$ if she can afford the whole path; otherwise she buys nothing.

\begin{lemma}
There is a polynomial-time algorithm that transforms an instance of Highway problem to an instance of $2$-{\sf SMP}.
\end{lemma}

\begin{proof}
For each $i=1,\ldots, n$, each edge $(v_{i-1}, v_i)$, we create an item $\Item_i$ at coordinates $(i, n+1-i)$. Then for each consumer $\Consumer$ whose path is $P_\Consumer = (v_s,\ldots, v_t)$, we create a consumer point at $(s+1, n+1-t)$. It is easy to see that the consideration set remains unchanged.
\end{proof}

\subsection{\APX-hardness of $4$-{\sf SMP}}

We perform a reduction from {\sf Graph Vertex Pricing} on bipartite graphs. In this problem, we are given a graph $G=(V,E)$, where each vertex corresponds to item and each edge $e \in E$ corresponds to a consumer, additionally equipped with budget $B_e$. Each consumer edge is interested in items that correspond to her incident vertices. Our goal is to set a price $p: V \rightarrow \R$ so as to maximize our revenue.

Given an instance $(G, \set{B_e}_{e \in E})$ of {\sf Graph Vertex Pricing} where graph $G$ is a bipartite graph $(U \cup W, E)$, we create an instance of $4$-{\sf SMP} as follows. Suppose we have $U = \set{u_1,\ldots, u_{|U|}}$ and $W = \set{w_1,\ldots, w_{|W|}}$. For each vertex $u_i \in U$, we have a corresponding item $\Item^u_{i}$ with coordinates $(i, |U|+1-i, \infty,\infty)$. Similarly, for each vertex $w_j \in W$, we have a corresponding item $\Item^w_{j}$ with coordinates $(\infty,\infty,j, |W|+1-j)$. Finally, for each edge $(u_i, w_j) \in E$, we have a consumer $\Consumer_{ij} = (i, |U|+1-i, j, |W|+1-j)$, whose budget is the same as the budget of edge $(u_i,w_j)$. The following claim is almost immediate.

\begin{claim}
For each consumer $\Consumer_{ij}$, we have that $S_{\Consumer_{ij}} = \set{\Item^u_i, \Item^w_j}$
\end{claim}

\begin{proof}
It is easy to see that $\set{\Item^u_i, \Item^w_j} \subseteq S_{\Consumer_{ij}}$. Notice that for $i'<i$, we have $\Consumer_{ij}[1] > \Item^u_{i'}[1]$, so any such item $\Item^u_{i'}$ cannot belong to $S_{\Consumer_{ij}}$. Similarly for $i' >i$, we have $\Consumer_{ij}[2] > \Item^u_{i'}[2]$, so such an item cannot belong to $S_{\Consumer_{ij}}$. By using similar arguments for items of the form $\Item^w_{j'}$ for $j' \neq j$, we reach the conclusion that $S_{\Consumer_{ij}} = \set{\Item^u_i, \Item^w_j}$.
\end{proof}

Since the $4$-{\sf SMP} instance is equivalent to the instance of {\sf Graph Vertex Pricing}, the maximum revenue is preserved. Using the \APX-hardness result of {\sf Graph Vertex Pricing} on bipartite graphs~\cite{KhandekarKMS09}, we conclude that $4$-{\sf SMP} is \APX-hard.

\subsection{Hardness Results in Higher Dimensions}\label{sec:higher
dimensions}

In this section, we present the proof of Theorem~\ref{theorem: higher dimension}. Let $A = (\iset, \cset)$ be an instance of {\sf UUDP-MIN} where every consumer ${\bf C}$ has its consideration set $S_\Consumer$ of size at most $B$. Let $\iset = \set{{\bf I}_1,\ldots, {\bf I}_n}$. For each $i \in [d]$, we pick a random permutation $\pi_i: [n] \rightarrow [n]$, so we have $d$ permutations $\pi_1,\ldots, \pi_d$. The function $\phi$ on items $\iset$ can be defined as $\phi({\bf I}_j)[i] = \pi_i(j)$, and we extend the function to the set of consumers as follows: $\phi({\bf C})[i] = \min_{j \in S_C} \pi_i(j)$. Now we have a well-defined function $\phi$.

\begin{lemma}
With probability at least $1- 1/n$, for all consumer ${\bf C} \in \cset$, the consideration set $S'_\Consumer$ defined by $S'_\Consumer = \set{{\bf I}_j: \phi({\bf I}_j) \mbox{ dominates } \phi({\bf C})}$ is exactly $S_\Consumer$.
\end{lemma}
\begin{proof}
Since we define $\phi({\bf C})$ to be the minimum of $\phi({\bf I}_j)$ over all items in $S_\Consumer$, we have $S_\Consumer \sse S'_\Consumer$. Let $k$ be the index of an item that does not belong to $S_\Consumer$. We show the following claim.

\begin{claim}
\label{claim: tiny prob}
The probability that $\phi({\bf I}_k)$ dominates $\phi({\bf C})$ is at most $1/n^{B+2}$.
\end{claim}
\begin{proof}
Fix some $i \in [d]$. The bad event that $\pi_i(k)\geq \min_{j \in
S_C} \pi_i(j)$ happens only if $\pi_i$ does not put $k$ in the last
position among $S_C \cup \set{k}$. This probability is exactly
$(1-1/(B+1))$. Therefore, the bad event happens for all values of
$i$ with probability at most $(1-1/(B+1))^d \leq 1/n^{B+2}$ for $d=
O(B^2 \log n)$.
\end{proof}

This claim immediately implies the lemma: By the union bound, the
probability that $\phi({\bf I}_k)$ dominates $\phi({\bf C})$ is at
most $1/n^{B+1}$. So we have that $\pr{}{S_\Consumer \neq S'_\Consumer} \leq
1/n^{B+2}$. There are at most $n^B$ possible consideration sets of
size at most $B$, so by union bounds, the probability that a bad
event $S_C \neq S'_\Consumer$ happens for some consumer ${\bf C}$ is at most
$1/n$.
\end{proof}

\paragraph{$n$ attributes capture general problem} Finally, we end this section with the proof that $n$-{\sf UUDP-MIN} captures the whole generality of {\sf UUDP-MIN}: Consider an instance $(\cset, \iset, \set{S_\Consumer}_{\Consumer\in \cset})$ of {\sf UUDP-MIN}. Denote the set of items by $\iset = \set{\Item_1,\ldots, \Item_n}$. Notice that we can define the coordinates of each consumer by $\Consumer[i] = 0$ if $\Item_i \in  S_\Consumer$, and $\Consumer[i] = 1$ otherwise. We define the coordinates of each item as $\Item_i[i] = 0$ and $\Item_i[j] =1$ for all $j \neq i$. It is easy to check that the consideration sets are preserved by this reduction.

\end{document}